\documentclass{article}
\usepackage{graphicx} 

\usepackage[letterpaper,top=2cm,bottom=2cm,left=3cm,right=3cm,marginparwidth=1.75cm]{geometry}

\usepackage{setspace} 
\usepackage{tabularx} 
\usepackage{fontspec}
\usepackage[english]{babel}
\usepackage{amsthm}
\usepackage{amsmath}
\usepackage{amssymb}
\usepackage{csquotes}
\usepackage{mathtools}
\usepackage[shortlabels]{enumitem}
\usepackage{dsfont}
\usepackage{float}
\usepackage{graphicx}
\usepackage[format=plain,
            labelfont=it,
            textfont=it]{caption}
\usepackage{subcaption}

\numberwithin{equation}{section}

\newcommand{\N}{\mathbb{N}}

\newcommand{\R}{\mathbb{R}}
\newcommand{\A}{\mathcal{A}}

\newcommand{\E}{\mathbb{E}}
\newcommand{\e}{\varepsilon}
\newcommand{\V}{\mathbb{V}}
\newcommand{\Cov}{\text{Cov}}
\newcommand{\Tr}{\text{Tr}}

\DeclareMathOperator*{\argmin}{arg\, min}

\newcommand{\as}{\xrightarrow{\text{a.s.}}}

\theoremstyle{definition}
\newtheorem{thm}{Theorem}[section]

\newtheorem{lemma}[thm]{Lemma}

\newcommand{\ny}{\medskip\noindent}

\usepackage{amsmath}
\usepackage{graphicx}
\usepackage[colorlinks=true, allcolors=blue]{hyperref}
\usepackage{natbib}

\title{
Choosing optimal Strang splitting estimators of nonlinear stochastic differential equation models\\ \, }

\author{Magnus Frederik Jensen$^{(1)*}$, Johan Ravn Cornelius$^{(1)*}$ and Susanne Ditlevsen$^{(1)}$\\
\vspace{2mm}\\
{\small (1) Department of Mathematical Sciences, University of Copenhagen, Copenhagen, Denmark}\\
{\small * Joint first authors.} \\
{\small \emph{Address for correspondence:} Susanne Ditlevsen, Department of Mathematical Sciences,}\\ {\small Universitetsparken 5, 2100 Copenhagen, Denmark. Email: susanne@math.ku.dk}}
\date{}

\begin{document}
\maketitle

\begin{abstract}
\noindent
The Strang splitting estimator is a powerful estimator for parametric inference in multivariate stochastic differential equation models with nonlinear drift and additive noise. While the choice of splitting does not affect the asymptotic distribution of the estimator, it makes a huge impact in finite-sample settings and it has not yet been shown how the splitting can be chosen optimally.
We derive error measures for the transition densities of the Strang splitting scheme, in particular calculating the bias up to the order of $h^3$, where $h$ is the length of the time step. 
We study the connection between these error measures and the performance of the Strang splitting estimator in the double-well potential model and the stochastic FitzHugh-Nagumo model, respectively. Our simulation studies suggest that linearization around fixed points yields accurate parameter estimates for potential models, while other splittings perform better for slow-fast excitable models.
\end{abstract}

\section{Introduction}

Statistical inference for nonlinear multivariate stochastic differential equation (SDE) models remains a major challenge in applied probability and statistics. Such models arise naturally in a wide range of scientific fields, including neuroscience, climate dynamics, ecology, finance, and biology, where complex dynamical behaviour is combined with stochastic forcing. In most practically relevant cases, the transition densities of the underlying diffusion processes are unavailable in closed form, rendering exact likelihood-based inference intractable. Consequently, statistical estimation must rely on approximations of the transition densities or on simulation-based methods. While a number of approaches have been proposed, many existing methods either become computationally demanding in high dimensions, require restrictive assumptions on the model structure, or lose accuracy in the presence of strong nonlinearities, see \cite{parameterestimation} and references therein for an overview of the literature.

A common strategy is to locally linearize the drift and apply Gaussian approximations, such as the Euler-Maruyama (EM) scheme and related methods. These approaches are often computationally efficient, but their accuracy can deteriorate substantially for highly nonlinear systems or for observation schemes with moderately large time steps. More accurate methods based on particle filtering or simulation-based likelihood approximations are available, but these methods may become prohibitively expensive in high-dimensional settings.

Recently, a new class of pseudo-likelihood estimators based on Strang splitting schemes was introduced for nonlinear SDEs with additive noise \citep{parameterestimation,pilipovic:2025}. The key idea is to decompose the drift field into a linear component combined with the diffusion term, which admits an analytically tractable Ornstein-Uhlenbeck (OU) flow, and a nonlinear deterministic component, which is solved separately. The resulting Strang splitting estimator combines computational simplicity with surprisingly strong numerical performance. In particular, the method scales well with dimension, remains computationally feasible for strongly nonlinear systems, and performs accurately even for relatively large observation intervals. The method preserves important structural properties of the underlying dynamics while avoiding many of the stability issues associated with standard discretization schemes \citep{parameterestimation}.

However, an important unresolved challenge is that the splitting itself is not unique. For a given nonlinear drift field, infinitely many decompositions into linear and nonlinear parts are possible, and different choices may lead to dramatically different finite-sample behaviour of the resulting estimator. Although asymptotically (for time step between observations going to zero and total observation interval going to infinity) all admissible splittings induce the same limiting distribution of the estimator \citep{parameterestimation}, their practical performance in finite samples may differ substantially in terms of bias, variance, numerical stability, and approximation accuracy of the transition densities. This raises two fundamental questions: how should an ”optimal” splitting be defined, and how can such a splitting be identified in practice?

In this paper, we investigate these questions systematically. Our focus is on understanding how the choice of splitting influences the finite-sample properties of the Strang splitting estimator and on developing principled criteria for selecting good splittings. Rather than relying solely on asymptotic arguments, we study local approximation errors and the accuracy of the approximating transition densities.

Our main contributions are the following. In Section \ref{section: moments of one-step predictions}, we derive higher-order moment expansions for the one-step predictions generated by the Strang splitting scheme and compare these with the moments of the true diffusion process. This yields explicit expressions for the leading-order bias and variance terms associated with a given splitting and provides theoretical insight into how different choices of linearization affect the quality of the approximation. In Section \ref{Section: Optimal Strang splitting}, we propose and investigate several notions of optimality for selecting splittings, including criteria based on minimizing local errors in expectation and variance of the transition density as well as global errors obtained by integrating over the invariant distribution.
We study both fixed and adaptive splitting schemes, demonstrating through simulation studies that the choice of splitting can have a substantial impact on estimation accuracy. Finally, we illustrate the methodology on one-dimensional and multidimensional nonlinear SDE models, namely the double-well potential model \citep{benzi1982stochastic,benzi1983theory} and the stochastic FitzHugh-Nagumo system \citep{FITZHUGH1961445,Nagumo1962,buckwar_etal}, where different splittings lead to markedly different finite-sample performance.

Overall, our results show that the quality of the Strang splitting estimators depends crucially on the chosen decomposition of the dynamics and that carefully designed adaptive splittings can significantly improve inference in nonlinear stochastic systems.

\section{Strang splitting estimator}\label{section: Strang splitting}
In this Section we restate the definition of the Strang splitting scheme and estimator from \cite{parameterestimation}.

\subsection{Problem setup}
We consider the following parametric model
\begin{align}\label{problem}
    d X_t = F_\beta(X_t)dt + \Sigma dW_t, \quad X_0=x_0,
\end{align}
where $X_t \in \mathcal{X} \subseteq \R^d$, $\beta \in \R^p,$ $\Sigma \in \R^{d\times d}$, $W$ is a $d$-dimensional Brownian motion and $F_\beta:\R^d\to\R^d$ is a function parametrized by $\beta.$ 
We denote the parameter of the model as $\theta=(\beta,\Sigma) \in \Theta \subseteq (\R^p \times \R^{d\times d})$ and consider discretely observed data of $X_t$ of the form $x_{0:N}=(x_{t_0},x_{t_1},\ldots,x_{t_N})\in \R^{d\times (N+1)}$ at time steps $0=t_0<\ldots<t_N=T$, which satisfy $t_n-t_{n-1}=h>0$ for all $n\in\{1,\ldots,N\}$. The goal is to estimate $\theta$ from $x_{0:N}$.

We impose the following assumptions, which ensures that \eqref{problem} has a unique strong solution (for details, see e.g. \cite{parameterestimation}). For every $\theta \in \Theta$,
\begin{enumerate}
   \item[\textbf{A1.}] $F_\beta$ is three times continously differentiable with respect to $x$ and twice continously differentiable with respect to $\theta$, i.e., $F_\beta \in C^3$ in $x$ and $F_\beta \in C^2$ in $\beta$.
   \item[\textbf{A2.}] $F_\beta$ is one-sided Lipschitz, i.e., there exists a constant $C>0$ such that $$(x-y)^\top (F_\beta (x) - F_\beta (y)) \leq C \| x-y\|^2 \quad \forall x,y \in \mathcal{X}.$$
   \item[\textbf{A3.}] $F_\beta$ and its partial derivatives are of polynomial growth in $x$, uniformly in $\beta$, i.e., there exist constants $C>0$ and $k \geq 1$ such that
   $$\| F_\beta (x) - F_\beta (y)\|^2 \leq C \left (1+\|x\|^{2k-2}+\|y\|^{2k-2} \right ) \|x-y\|^2 \quad \forall x,y \in \mathcal{X}.$$
   \item[\textbf{A4.}] $\Sigma\Sigma^{\top}$ is positive definite.
\end{enumerate}
Assumption {\bf A1} is stronger than needed; usually only $C^2$ in $x$ is required. However, we need the third derivative for the error evaluations of the pseudo-likelihood we propose later. Also assumption {\bf A4} is stronger than needed, as the results do generalize to hypoelliptic models (models with rank-defficient diffusion matrices, but which still admit densities with respect to the Lebesgue measure) but we will not treat that in this paper.

\subsection{The splitting scheme and the pseudo-likelihood}
Strang splitting relies on the splitting of $F(x)$ in (\ref{problem}) into a linear part $A(x-b)$ and a nonlinear part $N(x)$, such that $F(x) = A(x-b)+N(x)$. Such splittings always exist, since for any $A$ and $b$, we define $N(x) = F(x) - A(x-b)$. We construct the following subequations for a given $A\in \R^{d\times d}$ and $b\in \R^d:$

\begin{align}
    &dX_t^{[1]}=A\left(X_t^{[1]}-b\right)\,dt +\Sigma\,dW_t, &X_0^{[1]}=x_0, \label{linear split}\\
&dX_t^{[2]}=N\left(X_t^{[2]}\right)\,dt, &X_0^{[2]}=x_0. \label{nonlinear split}
\end{align}
Denote the $h-$flows of these two systems by $\Phi_h$ and $f_h$, respectively, i.e., the forward propagation of the processes defined by \eqref{linear split} and \eqref{nonlinear split} from time $t$ to $t+h$, conditional on $X_t^{[j]}, j = 1,2$, 
\begin{align*}
    X_{t+h}^{[1]} &= \Phi_h\left(X^{[1]}_t\right),\\
    X_{t+h}^{[2]} &= f_h\left(X_t^{[2]}\right).
\end{align*}
The solution to (\ref{linear split}) is an OU process, and its $h-$flow is a random variable given by
\begin{align*}
    \Phi_h\left(x\right)=e^{Ah}x+(I-e^{Ah})b+\xi_{h},
\end{align*}
where $\xi_{h}\sim\mathcal{N}(0,\Omega_h)$ with 
\begin{equation}
\label{eq: Omega_h}
 \Omega_h=\int_0^he^{A(h-r)}\Sigma\Sigma^{T}e^{A^{\top}(h-r)}dr
= h\Sigma\Sigma^{\top} + \frac{h^2}{2}(A\Sigma\Sigma^{\top}+\Sigma\Sigma^{\top}A)+O(h^3).   
\end{equation}

\ny
Denoting $\Phi_h^{[S]}=f_\frac{h}{2}\circ \Phi_h \circ f_\frac{h}{2},$
the Strang scheme is iteratively defined by
\begin{align}\label{eq: Strang flow}
    X^{[S]}_{t_k} 
    &=
    \Phi_h^{[S]}\left(X^{[S]}_{t_{k-1}}\right) 
    = f_{\frac{h}{2}} \left( \mu_h \big( f_\frac{h}{2}(X^{[S]}_{t_{k-1}})\big) + \xi_{h,k} 
    \right),
\end{align}
where $\mu_h(x)=e^{Ah}x+(I-e^{Ah})b.$
We refer to $\Phi_h^{[S]}(X_{t_{k-1}})$ as the \textit{one-step prediction} of $X_{t_k}$.

\ny
The Strang scheme induces a pseudo-likelihood with transition densities that are nonlinear transformations of Gaussians. The pseudo-loglikelihood is given (up to a constant and proportionality) by
\begin{align}
\ell^{[S]}_N(\theta)
    &= 
    -N \log\big( \det \Omega_h\big)
    + \sum_{k=1}^N
    \left(- z_{t_k}^\top \Omega_h^{-1} z_{t_k} + 2\log\left( | \det D f_{\frac{h}{2}}^{-1}(x_{t_k}) |\right) \right), \label{eq:loglikelihood}
\end{align}
where
    $
z_{t_k}=f^{-1}_{\frac{h}{2}}(x_{t_k}) - \mu_h(  f_{\frac{h}{2}}(x_{t_{k-1}})),
    $
yielding the Strang splitting estimator $\hat \theta$ as the value of $\theta$ that maximizes the pseudo-loglikelihood. The first two terms are the usual terms in a Gaussian likelihood, the last term is due to the nonlinear transformation.

In \eqref{eq:loglikelihood}, we implicitly assume that the choice of $b$ and $A$ and thus also $N$ in the splitting is the same for all one-step predictions. However, we can also define an adaptive splitting scheme where the linear part (\ref{linear split}) (and thus $N(x)$) is chosen in some optimal way as a function of the data. To construct a certain transition density of $x_{t_k}$ given $x_{t_{k-1}}$, we might wish to choose $b$ and $A$ in a way that depends either on $x_{t_k}$ or on $x_{t_{k-1}}$, and we shall see examples of such strategies in Section \ref{Section: Optimal Strang splitting}.
Let $\Omega_{h;x}, \mu_{h;x}$ and $ f_{h;x}$ denote the replacements of the corresponding expressions in \eqref{eq:loglikelihood} when $A,b$ and $N$ depend on the data $x_{0:N}$. With $z_{t_k;x}=f_{\frac{h}{2};x}^{-1}(x_{t_k})-\mu_{h;x}(f_{\frac{h}{2};x}(x_{t_{k-1}})),$ we then obtain 
\begin{align}\label{eq: adaptive splitting likelihood}
    \ell_N^{[S]}(\theta)
    = \sum_{k=1}^N
    \left(
    -\log \left( \det \Omega_{h;x} \right)
    -  z_{t_k;x}^\top \Omega_{h,x}^{-1}
   z_{t_k;x}
    +2 \log \left( | \det D f_{\frac{h}{2};x}^{-1} (x_{t_k}) | \right)
    \right)
\end{align}
as the pseudo-loglikelihood for the adaptive splitting scheme.

\section{Moments of one-step predictions}\label{section: moments of one-step predictions}

It was shown in \cite{parameterestimation} that 
\begin{align*}
 \E\left[ X_{t_k}-\Phi^{[S]}(X_{t_{k-1}}) \ \big| \  X_{t_{k-1}}=x\right]
 =O(h^3)
\end{align*}
for $x\in\R^d$ for all choices of $b$ and $A$. 
The main theoretical contribution in this paper is Theorem $\ref{nyt teorem}$ which extends this result. We prove that the Strang one-step prediction is in fact biased in the order $h^3$. Additionally, there is an $h^3$-order error in the variance, and both terms depend on the choice of splitting. 
This may offer an explanation as to why the Strang splitting estimator performs better for certain splittings even though the asymptotic distribution of the estimator is the same for all choices.

For notational simplicity, in this Section we restrict ourselves to one-dimensional SDEs, i.e., $d=1$, and suppress the dependence on the parameter in the notation. We expand the moments of $X_{t_k}$ using the following lemma \citep[Lemma 1.10]{EstimatingFunctions}, adapted to our setting. 
\begin{lemma}\label{lemma infinitesimal generator}
    Let $X$ be the solution to \eqref{problem} 
    satisfying  assumptions \textbf{A2-A3}
    and let $\mathcal L $ be its infinitesimal generator,
    \[
(\mathcal L f)(x)
=F(x)\,\frac{d f}{d x}(x)
+
\frac12
\Sigma^2
\frac{d^2 f}{d x^2}(x).
\]
Assume that $F \in C^{2n}(\mathcal{X})$ in $x$ for some $n\in\N$ and that $F$ and all its derivatives are of polynomial growth. Assume that $\varphi:\R\to\R$ is $C^{2(n+1)}(\mathcal{X})$ and that $\varphi$ and its derivatives are also of polynomial growth. Then, for all $t,h\geq 0$ and $x\in\R$, 
    \begin{align*}
        \E\left[ \varphi(X_{t+h}) \, | \, X_{t} =x \right]
    = \sum_{j=0}^n \frac{h^j}{j!} (\mathcal L^j\varphi)(x)
    + O(h^{n+1}),
    \end{align*}
    where $\mathcal L^j$ denotes the operator $\mathcal L$ applied $j$ times.
\end{lemma}
 
\ny
We now state our main result, the proof of which is found in Section \ref{section: proofs}.

\begin{thm}\label{nyt teorem}
        Let $dX_t=F(X_t)dt+\sigma dW_t$ be a one-dimensional SDE with $F \in C^5(\mathcal{X})$ satisfying assumptions \textbf{A1}-\textbf{A4} and let $\Phi_h^{[S]}$ be the Strang $h$-flow defined in \eqref{eq: Strang flow}. Then
    \begin{enumerate}[label=(\roman*)]
        \item The bias of $\Phi_h^{[S]}(X_{t_{k-1}})$ is of order $h^3$ and depends on the splitting, such that\\ 
        $\E[ \Phi_h^{[S]}(X_{t_{k-1}})- X_{t_k}\,|\, X_{t_{k-1}}=x ]=h^3 B_\theta(x)+O(h^4)$ where
\begin{eqnarray}
    B_\theta(x)  
    &=& \nonumber 
    \frac{1}{24}\Big(AN'(x)N(x)+N''(x)N(x)A(x-b) -N'(x)^2A(x-b)\Big)\\
    &&
     +\frac{1}{12}\Big(A^2N(x)-N'(x)A^2(x-b)+N''(x)A^2(x-b)^2\Big)\\ 
    &&
     +\sigma^2\Big(\frac{1}{12}N'''(x)A(x-b)-\frac{1}{16}N''(x)N'(x)+\frac{1}{48}N'''(x)N(x)\Big) + \frac{\sigma^4}{48}N''''(x)\nonumber
    . \label{eq:B}
\end{eqnarray}

        \item The difference in variance of $\Phi_h^{[S]}(X_{t_{k-1}})$ is of order $h^3$ and depends on the splitting, such that\\
        $\V[\Phi_h^{[S]}(X_{t_{k-1}})|X_{t_{k-1}}=x]-\V[X_{t_k}|X_{t_{k-1}}=x]=h^3 \Delta V_\theta(x)+O(h^4)$ where
        \begin{eqnarray}
    \Delta V_\theta(x)  
    &=& \nonumber 
    \sigma^2\left(\frac{1}{12}N''(x)N(x) 
    + \frac{1}{3}N''(x)A(x-b)
    -\frac{1}{6}N'(x)^2
    -\frac{1}{3}AN'(x)\right)+\frac{\sigma^4}{6}N'''(x) .
\end{eqnarray}
    \end{enumerate}
\end{thm}
\noindent

\section{Optimal Strang splitting}\label{Section: Optimal Strang splitting}
We (loosely) define a splitting to be optimal if it yields more accurate parameter estimates than all other choices of $A$ and $b$. Since the Strang pseudo-likelihood is a product of transition densities, we expect a splitting scheme that accurately approximates these transition densities to also provide an accurate estimator. Hence, we search for an optimal Strang splitting by studying the two error functions  $B_\theta(x)$ and $\Delta V_\theta(x)$ of Theorem \ref{nyt teorem}, which depend on $A$ and $b$, for a given model.

A natural and promising splitting in the general $d$-dimensional case is when the drift of the SDE in question has a fixed point, i.e., $x^*$ such that $F(x^*)=0$. We can then choose $A(X_t^{[1]}-b)$ in (\ref{linear split}) to be the first-order Taylor expansion of $F$ around $x^*$ by letting $A=DF(x^*)$ and $b=x^*$. Close to $x^*$, the dynamics will be dominated by the linear OU process and the pseudo-likelihood will be close to the true likelihood. This splitting was used in \cite{ditlevsen:2023, ditlevsen:2026} for the saddlenode-bifurcation model to estimate time to tipping in climate. 

If the eigenvalues of $DF(x^*)$ have negative real parts, the fixed point $x^*$ is stable, and the OU process is also mean-reverting towards $x^*$. Intuitively, the process will often be near the fixed point, making the Taylor linearization a good approximation.
If the model has multiple stable fixed points, each one locally acts as an attractor for initial values of the process within some neighbourhood of the fixed point called the basin of attraction. An adaptive splitting \eqref{eq: adaptive splitting likelihood} can then be constructed by linearizing around the fixed point that is the attractor for each $x_{t_k}.$ A similar splitting, where the analytic stable fixed points were substituted by the empirical means within basins of attraction, was applied in \cite{pilipovic:2025} for Kramers oscillator fitted to calcium ion concentrations in ice cores from Greenland.

For certain systems, for example systems exhibiting deterministic chaos, the splitting can also be done around unstable fixed points, in which case the OU process is repelled away from $x^*$. This splitting was used in \cite{parameterestimation} for the Lorenz system, where all fixed points are unstable. 

We will refer to splittings around fixed points as the \textit{fixed point linearization} scheme. 
However, the approximation far from the fixed point, as we will see, will for certain types of models be poor, in particular for models where the process spends extended time far from the equilibrium, such as excitable models. 
More generally, we linearize around  $c \in \R^d$
by setting $A=DF(c)$ and $b=c-A^{-1}F(c),$ assuming that $A$ is invertible, such that 
the drift of the OU becomes $A(X_t^{[1]}-b)=F(c)+DF(c)(X^{[1]}_t-c)$. For $d=1$, this makes for a natural study of the bias and variance functions of Theorem \ref{nyt teorem} as mappings that depend on
the centre of linearization $c\in\R$:
\begin{align}
    &(x,c)\mapsto B_\theta(x,c)   \label{eq: bias function}\\
    &(x,c)\mapsto \Delta V_\theta(x,c) \label{eq: variance function}
\end{align}
We are interested in the question of how $c$ can be chosen to minimize $|B_\theta(x,c)|$ and $| \Delta V_\theta(x,c) |$, respectively, for different $x$ in state space.

\subsection{Model classes}

We consider two model classes, which we will treat separately. First, we present evidence for a general strategy for splitting in one-dimensional \textit{potential models}, using the fixed point linearization scheme, and then we consider a two-dimensional \text{slow-fast} system - a case where the same fixed point linearization scheme is not the optimal choice. 

\paragraph{Potential models} We consider the multidimensional gradient SDE with $F_\beta = -\nabla U_\beta \in \mathbb{R}^{d}
$, such that
\[
dX_t = -\nabla U_\beta(X_t)\,dt + \Sigma  \, dW_t,
\qquad X_t\in\mathbb{R}^d,
\]
where \(U_\beta:\mathbb{R}^d\to\mathbb{R}\) is a smooth potential, $\Sigma \in\R^{d\times d}$, and the gradient operator is given by $(\nabla U_\beta(X_t))_{i} := \partial (U_\beta(x))/\partial x_i \, \Big |_{x=X_t} , i = 1, \ldots , d$. 

Assume that \(U_\beta\in C^2(\mathbb{R}^d)\) and that \(U_\beta\) is confining in the sense that
\[
U_\beta(x)\to\infty,
\qquad \text{as } |x|\to\infty,
\]
sufficiently fast so that
\[
\int_{\mathbb{R}^d}
\exp\left(-\frac{2U_\beta(x)}{\sigma^2}\right)\,dx < \infty.
\]
Then the process admits the invariant density \citep{thygesen}, \citep[Theorem 5.6]{applied}
\begin{equation}
\label{eq:invariantgeneral}
\pi(x)=\frac{1}{Z}
\exp\left(-\frac{2U_\beta(x)}{\sigma^2}\right),
\qquad 
Z=\int_{\mathbb{R}^d}
\exp\left(-\frac{2U_\beta(x)}{\sigma^2}\right)\,dx .
\end{equation}
Thus, local minima of \(U_\beta\) correspond to modes of the stationary distribution. The potential landscape determines metastability, tipping behaviour (i.e., bifurcations or noise-induced transitions between attractors), and transition paths.

In general, not every multidimensional SDE admits a potential formulation. A necessary condition is that the drift field is conservative:
\[
\frac{\partial (F_\beta)_i}{\partial x_j}
=
\frac{\partial (F_\beta)_j}{\partial x_i},
\qquad i\neq j,
\]
for all $i,j=1, \ldots , d$ (assuming sufficient smoothness and simply connected domain). Systems violating this contain rotational/non-gradient components and cannot globally be represented by a scalar potential.

\paragraph{Slow-fast excitable models} 
A general slow--fast excitable SDE with additive noise can be written as
\[
\begin{aligned}
dX_t &= \frac{1}{\varepsilon} f_\beta(X_t)\,dt
      + \frac{1}{\sqrt{\varepsilon}}\Sigma_X\,dW_t^{X},\\
dY_t &= g_\beta(X_t)\,dt
      + \Sigma_Y\,dW_t^{Y},
\end{aligned}
\qquad 0<\varepsilon\ll 1,
\]
where $(X_t^\top, Y_t^\top)^\top$ is the full state process, $F_\beta = (f_\beta^\top, g_\beta^\top)^\top$, $X_t \in \mathbb{R}^{d_x}$ denotes fast excitatory variables, $Y_t \in \mathbb{R}^{d_y}$ denotes slow recovery/adaptation variables, $d=d_x+d_y$, and $\varepsilon$ separates the time scales. The deterministic drift is assumed to have an excitable geometry: a stable fixed point coexists with a threshold manifold such that small perturbations across this threshold might generate a large transient excursion in state space before returning to the vicinity of the stable fixed point.

\subsection{A potential model: the double-well model}\label{subsection double-well}

The \textit{double-well potential  model} is the one-dimensional SDE with $U(x) = \frac{1}{\varepsilon}(\frac14 x^4 - \frac12 x^2 + xy)$, 
\begin{eqnarray}
\label{eq:DW}
    dX_t &=& \frac{1}{\e}(X_t-X_t^3-y)\, dt + \sigma \, dW_t,
    \quad X_0=x_0
\end{eqnarray}
for $\e,\sigma>0,y\in\R$. For $|y| < 2/3\sqrt{3}$, \eqref{eq:DW} has two stable fixed points separated by an unstable fixed point.
Figure \ref{Plots double-well model} shows the potential function $U$ as well as sample paths from the corresponding model \eqref{eq:DW} for different parameter values. The potential function is symmetric for $y=0$, whereas for $y\neq 0$, one of the two potential wells becomes deeper, making the corresponding stable fixed point the stronger attractor. 

\begin{figure}[h]
\begin{subfigure}[b]{0.49\linewidth}
    \includegraphics[height=7 cm]{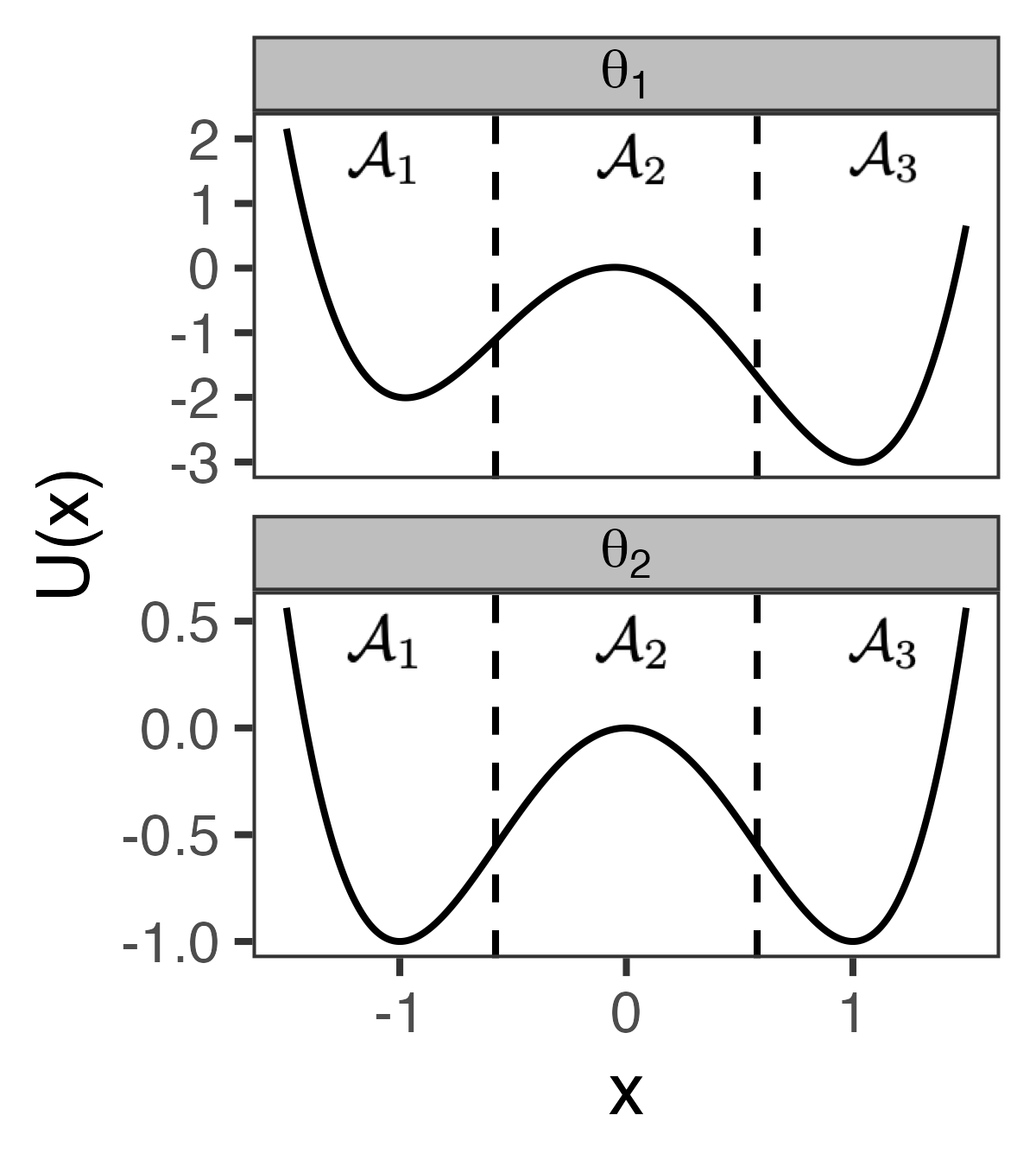}
 \end{subfigure}
 \hspace{-1.5 cm}
 \begin{subfigure}[b]{0.49\linewidth}
\includegraphics[height=7 cm]{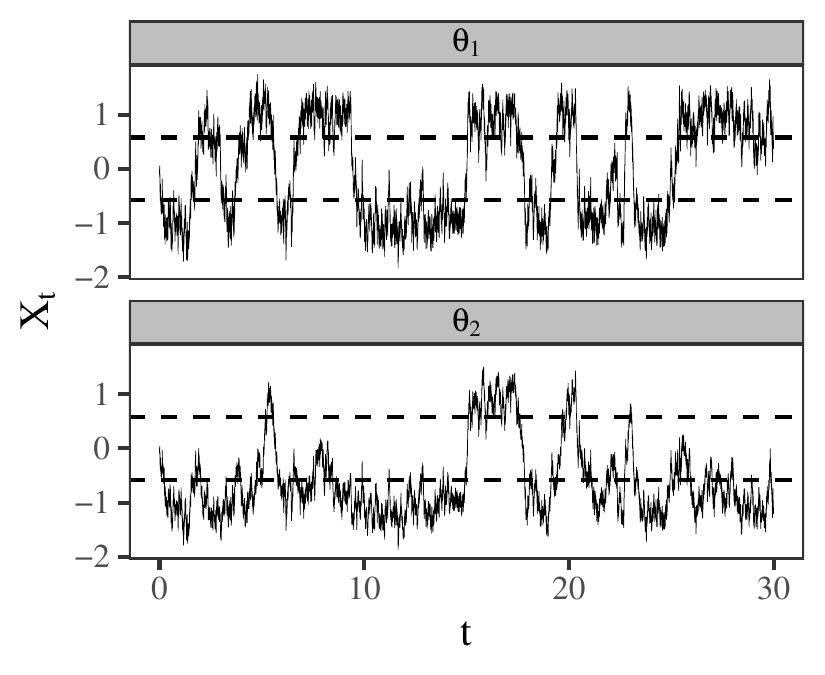}
 \end{subfigure}
\caption{\textbf{Potential functions and sample paths for the double-well potential model.} Parameters are $\theta_1=(\e_1,y_1,\sigma_1)=(0.1, -0.05, 1.7)$ and $\theta_2=(\e_2,y_2,\sigma_2)=(0.25, 0, 1.2)$.  
Both sample paths are simulated with the EM scheme using $h=0.0001$ and the same realization of the underlying Brownian motion. Vertical  dashed lines partition the state space into the three regions defined in \eqref{eq: partitioned state space}.}
    \label{Plots double-well model}
\end{figure}

Linearizing the drift in \eqref{eq:DW} around $c\neq \pm \frac{1}{\sqrt{3}}$ yields a well-defined drift in \eqref{linear split} given by
\begin{align*}
    &A = \frac{1}{\e}(1-3c^2),\\
    &b = c- \frac{c-c^3-y}{1-3c^2},
    \\
    &N(X_t) =
\frac{1}{\e}(-X_t^3+3c^2X_t-2c^3).
\end{align*}
For this model, we suggest the following fixed point linearization scheme. Specifically, denoting $x^-,x^\dagger$ and $x^+$ the negative stable, unstable, and positive stable fixed points, respectively, we use the splitting 
\begin{equation}
c=
\begin{cases}
x^-        & \text{if } X_t \in \mathcal{A}_1 \\
x^\dagger  & \text{if } X_t \in \mathcal{A}_2 \\
x^+        & \text{if } X_t \in \mathcal{A}_3
\end{cases}
\label{eq:fixed point scheme}
\end{equation}
where 
\begin{equation}\label{eq: partitioned state space}
\mathcal{A}_1=\left(-\infty,\frac{-1}{\sqrt{3}}\right],\,
\mathcal{A}_2=\left(\frac{-1}{\sqrt{3}},\frac{1}{\sqrt{3}}\right)
\text{ and }\mathcal{A}_3=\left[\frac{1}{\sqrt{3}},\infty\right).
\end{equation}
This is a natural choice as $\mathcal{A}_1$, $\mathcal{A}_2$ and $\mathcal{A}_3$ partition the state space into regions where the drift $F(x)$ is monotone, as $F'(x)$ vanishes at $x=\pm\frac{1}{\sqrt{3}}$. Thus, the direction of flow of the OU process, which is determined by the sign of $A$, always aligns with the direction of the true local flow.

\paragraph{One-step predictions} Theorem \ref{nyt teorem} states the presence of bias and error in variance of the Strang one-step prediction and that these depend on the choice of splitting (in addition to parameters, $h$ and the state of the process). To see that this is actually the case in practice, we compare empirical densities of one-step predictions of size $h=0.08$ from three different choices of linearization schemes in Figure \ref{plot: transition densities}. 
To assess the accuracy of the densities, we accumulate 2000 sub steps for each step $h$ from the EM scheme such that it is reasonable to assume that the result resembles the true distribution, given by the yellow curves. The red curve, which is the distribution from the fixed point linearization scheme defined in $\eqref{eq:fixed point scheme}$, seems very accurate. In the blue curve we instead use $x^+$ when $x_0\in\mathcal{A}_1$ and $x^-$ when $x_0\in\mathcal{A}_3$, which results in a very poor approximation of the transition density.

We also include a local Strang splitting that uses $c=x_0$. This seems like at natural choice for obtaining the best local linearization of the drift, but it carries an unintended consequence. Linearizing around $x_0$ means that we force $N(x_0)=0$ so the Strang composition will skip the first nonlinear step, reducing \eqref{eq: Strang flow} to $\Phi_h^{[S]}(X^{[S]}_{t_{k-1}}) 
    = f_{\frac{h}{2}} ( \mu_h ( X^{[S]}_{t_{k-1}}) + \xi_{h,k} 
    )$. 
The tails of the corresponding transition density, the green curve in Figure \ref{plot: transition densities}, do not align well with those of the true density, though the approximation of the mode is accurate.

\begin{figure}[h]
    \centering
     \includegraphics[height=7.5 cm]{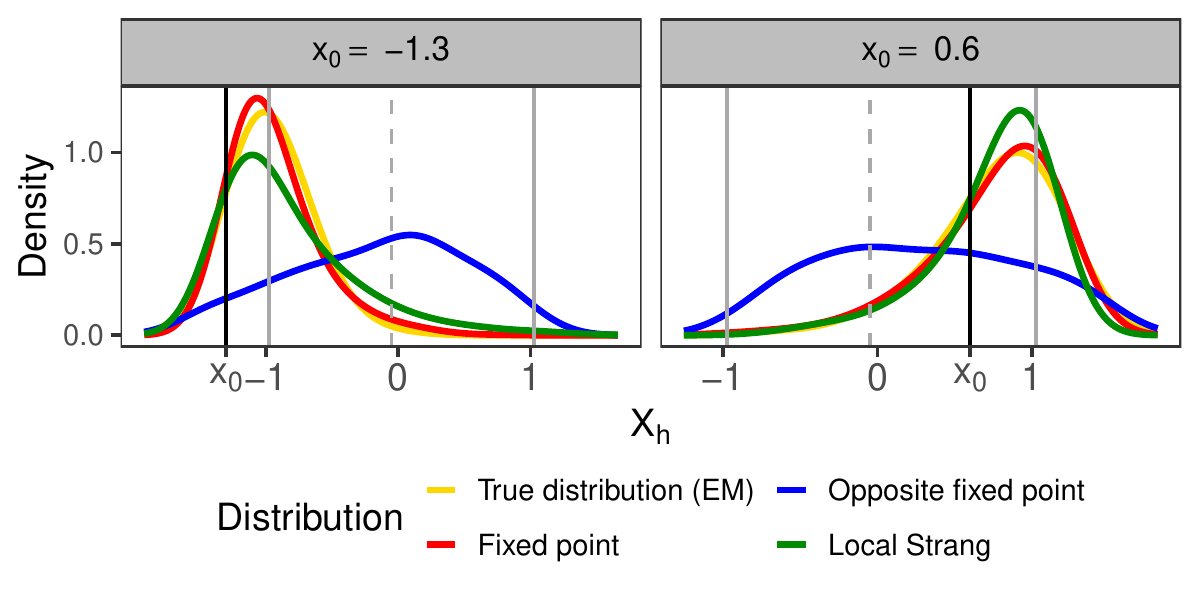}
    \caption{\textbf{Estimated transition densities of three Strang splitting schemes in the double-well potential model.}
    For each density, 2000 values with $h=0.08$ and the parameters $y=-0.05$, $\e=0.1$ and $\sigma=1.5$ are simulated, and a different initial value $x_0$ is used for each panel.
    The true distribution is simulated from the EM scheme with $h^{sim}=0.00004$. The fixed point approximation is done by linearizing around the fixed point of the same well as $x_0$, opposite fixed point approximation by linearizing around the fixed point of the other well, and local Strang is linearizing around $x_0$.
    Solid (resp. dashed) grey lines indicate stable (resp. unstable) fixed points.}
    \label{plot: transition densities}
\end{figure}

\ny 
To see more generally how the choice of linearization affects the Strang transition density, we will now turn our focus to  $B_\theta(x,c)$ and $\Delta V_\theta(x,c)$ from Theorem \ref{nyt teorem}. We plot the functions against each other for varying $x$ and $c$ in Figure \ref{figure: bias-variance tradeoff}. The curves show how these two error functions develop as the centre of linearization $c\in[-1.5,1.5]$ varies continuously. The $c$-values corresponding to the three fixed points as well as the minimizers of each function are marked with coloured dots.

\begin{figure}[h!]
    \centering

    \begin{subfigure}{\linewidth}
        \centering
        \includegraphics[width=\linewidth]{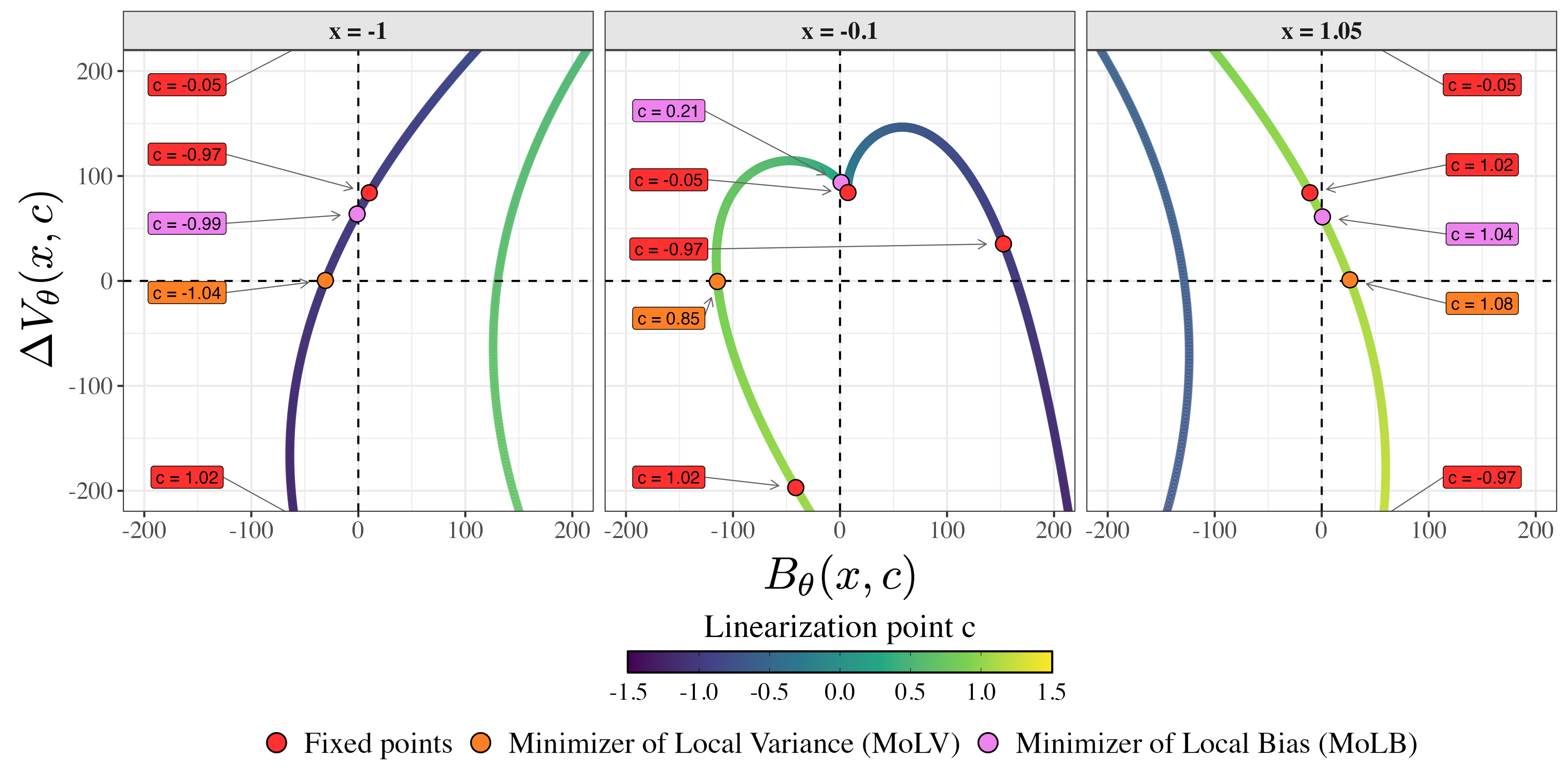}
        \caption{\textbf{Local errors.}
    The $h^3$-order bias $B_\theta(x,c)$ and difference in variance $\Delta V_\theta(x,c)$, both defined in Theorem \ref{nyt teorem}, are shown for continuously varying $c$. Each panel illustrates a different initial value of $x$.  
    }
    \label{figure: bias-variance tradeoff}
    \end{subfigure}

    \vspace{0.5cm}

    \begin{subfigure}{\linewidth}
        \centering
        \includegraphics[width=\linewidth]{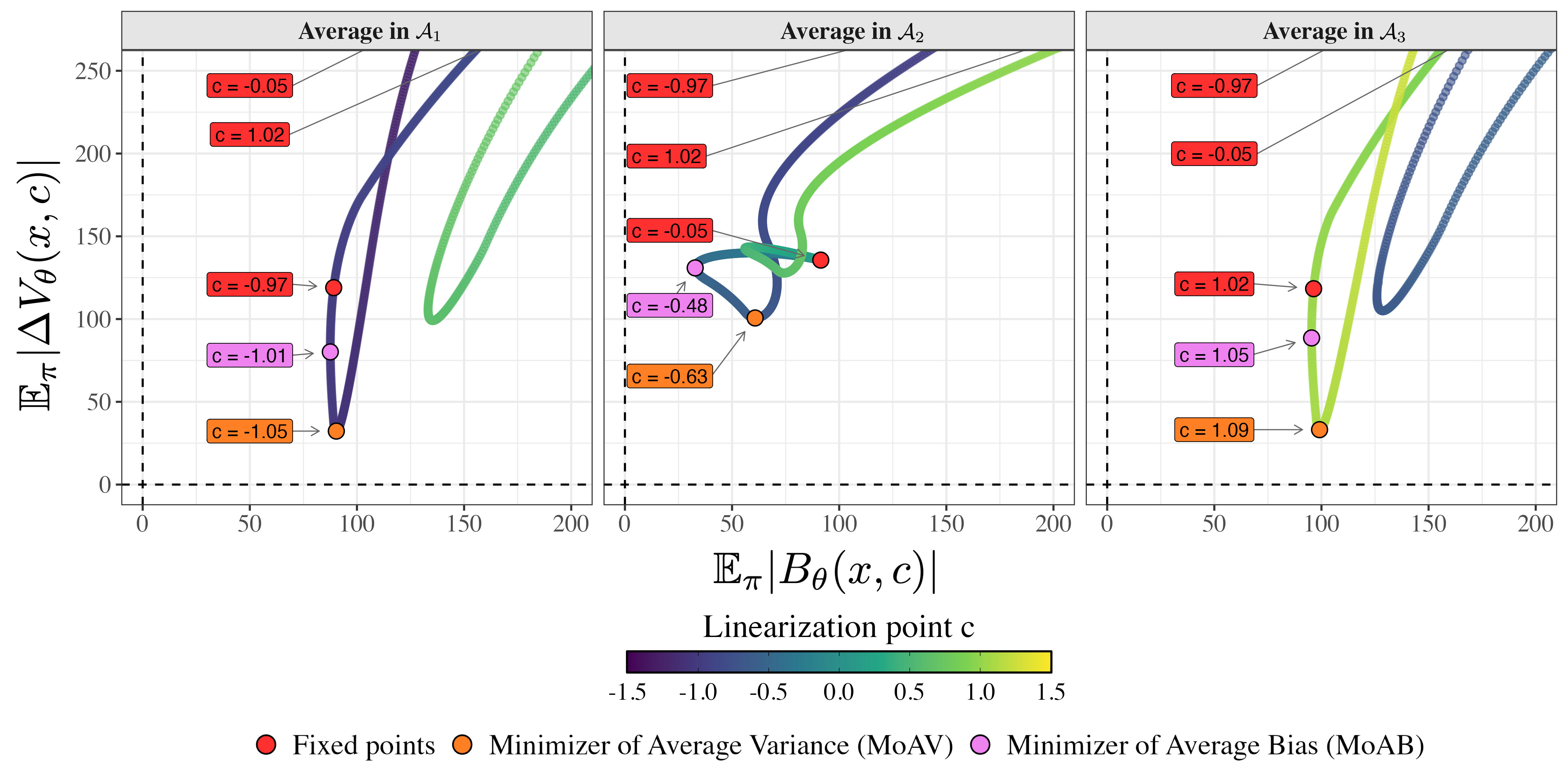}
        \caption{\textbf{Mean errors.}
    Expected values of $|B_\theta(X,c)|$ and $|\Delta V_\theta(X,c)|$ with respect to the invariant density $\pi$, conditional on the events $X\in\mathcal{A}_i$ for $i\in\{1,2,3\}$, are shown for continuously varying $c$.  
    }
    \label{figure: integrated bias-variance tradeoff}
    \end{subfigure}

    \caption{\textbf{Errors in expectation and variance of Strang one-step predictions for the double-well potential model.} Fixed points and minimizers are highlighted with coloured dots. The colour indicates the value of $c$. Parameters are $\e=0.1$, $y=-0.05$, and  $\sigma=1.7$.
    }
    \label{fig:strang_estimators}
\end{figure}

Three initial values of $x$ are chosen within each of the regions $\mathcal{A}_1$, $\mathcal{A}_2$ and $\mathcal{A}_3$ such that they are close to their respective fixed point. In each case, it can immediately be seen that $c$ chosen to be the local fixed point is very close to unbiased, indicated by the vertical dashed lines. This is not a coincidence, as we in fact have for all $\theta$ that if $x^*$ is any fixed point, then $B_\theta(x^*,x^*)=0$. Indeed, the bias function in \ref{nyt teorem}(i) reduces to $B_\theta(x^*,x^*)=\frac{\sigma^4}{48}N''''(x^*)-\frac{\sigma^2}{16}N''(x^*)N'(x^*)$. Specifically in the double-well model, both terms vanish, since $N'(x)=\frac{1}{\e}(3c^2-3x^2)$ and $N''''(x)=0$. 

\ny
Rather than evaluating in specific $x$-values in state space, it is of more interest to understand how $B_\theta(X_t,c)$ and $\Delta V_\theta (X_t,c)$ behave on average inside the three regions $\mathcal{A}_1$, $\mathcal{A}_2$ and $\mathcal{A}_3$ defined in \eqref{eq: partitioned state space}. We achieve this by utilizing that the double-well model is a potential model and thus has an invariant distribution given in \eqref{eq:invariantgeneral} with density
\begin{align}
\label{eq: invariant distribution}
    \pi(x) =
    \frac{1}{Z}\exp\left( \frac{1}{\e\sigma^2}\left(-\frac{1}{2}  x^4  + x^2 -2xy
    \right)\right).
\end{align}
By Theorem 1.15 of \cite{kutoyants}, the double-well potential model is ergodic so
\begin{align*}
        \frac{1}{T}\int_0^T |B_\theta(X_t,c)|\, dt \as \E( |B_\theta(X,c)| )\quad\text{and}\quad \frac{1}{T}\int_0^T |\Delta V_\theta(X_t,c)|\, dt \as \E( |\Delta V_\theta(X,c)| ),
    \end{align*}
as $T\to\infty$, where $X\sim \pi $ is a random variable with invariant density $\pi$, as given in \eqref{eq: invariant distribution}.
For every $c$ we can thus find the mean value of $|B_\theta(X_t,c)|$ and $|\Delta V_\theta(X_t,c)|$ across any sufficiently long trajectory $(X_t)_{t\in[0,T]}$ by computing $\int_\R |B_\theta(x,c)| \pi(x) dx$ and $\int_\R |\Delta V_\theta(x,c)| \pi(x) dx$.

For $i\in\{1,2,3\}$ we then condition on the event $X\in\mathcal{A}_i$ and consider mean values of $|B_\theta(X,c)|$ and $|\Delta V_\theta(X,c)|$ as functions of the centre of linearization $c$, i.e.,
\begin{align*}
    c\mapsto\E_\pi\left( |B_\theta(X,c)|  \mid X\in \mathcal{A}_i \right)
    &=\frac{1}{\int 1_{\mathcal{A}_i}(x)\pi(x)\, dx}
    \int 1_{\mathcal{A}_i}(x) |B_\theta(x,c)| \pi(x)\, dx
\end{align*}
and 
\begin{align*}
    c\mapsto\E_\pi\left( |\Delta V_\theta(X,c)|  \mid X\in \mathcal{A}_i \right)
    &=\frac{1}{\int 1_{\mathcal{A}_i}(x)\pi(x)\, dx}
    \int 1_{\mathcal{A}_i}(x) |\Delta V_\theta(x,c)| \pi(x)\, dx
    .
\end{align*}

\noindent
In Figure \ref{figure: integrated bias-variance tradeoff}, we plot these mappings against each other, again for continuously varying $c\in[-1.5,1.5]$. 
In the left and right panels, it is clearly seen that the best choices of $c$ with respect to both bias and variance of the Strang one-step predictions are found in the same region as the current value of the process. In the middle panel, negative $c$-values are slightly better than the positive - likely due to the negative value of the control parameter $y=-0.05$. Though the fixed points are not exactly optimal in either of the three events, the stable fixed points are close to the values that minimize the mean error functions within $\A_1$ and $\A_3$, while the unstable fixed point is a better choice in the event $X_t \in \A_2$. 
Thus,
Figures  \ref{plot: transition densities},  \ref{figure: bias-variance tradeoff} and \ref{figure: integrated bias-variance tradeoff}  all serve as evidence that the fixed point linearization scheme defined in \eqref{eq:fixed point scheme}
yields accurate transition densities in terms of expectation and variance. 

It is interesting to examine whether choosing the different minimizers in Figure \ref{figure: bias-variance tradeoff} and \ref{figure: integrated bias-variance tradeoff} will benefit the estimator. We implement these in a simulation study (see \ref{eq: MoLB}) but is should be noted that regardless of their accuracy, the usefulness of these methods will be limited in practice due to computational complexity and the fact that they are tedious to implement. Even if analytical expressions for the minima of \eqref{eq: bias function} and \eqref{eq: variance function} were available, adapting the splitting in every step does not scale well with higher dimensions, mainly due to the many recalculations of the inverse covariance matrix $\Omega_{h;x}$ in \eqref{eq: adaptive splitting likelihood}.

\paragraph{Estimators}
Motivated by the preceding analysis of one-step predictions, we will now in detail present the Strang estimators to be compared on their accuracy in parameter estimation. Once the splitting is chosen, we define the estimator as the minimizer of the pseudo-loglikelihood \eqref{eq: adaptive splitting likelihood}, and as we only consider Taylor linearization schemes, all estimators are uniquely determined by the choice of $c$.  The following schemes are implemented:

\ny
\underline{Piecewise schemes} where three centres of linearizations $c_1,c_2$ and $c_3$ are used in each of the regions $\mathcal{A}_1$, $\mathcal{A}_2$ and $\mathcal{A}_3$ defined in \eqref{eq: partitioned state space}. 

\ny
\begin{itemize}
    \item \textit{Fixed point linearization} as defined in \eqref{eq:fixed point scheme} and \textit{Wrong fixed point linearization} where $x^-$ and $x^+$ are interchanged.

    \item \textit{Minimizer of Average Bias }(MoAB) and \textit{Minimizer of Average Variance} (MoAV) defined, respectively, as 
    \begin{align}\label{eq: MoAB}
    &c_i\in\argmin_{c\in\R\setminus \{\pm \frac{1}{\sqrt{3}}\}} \E_\pi(|B_\theta(X,c)|\mid X\in\mathcal{A}_i)\qquad\text{for }i\in\{1,2,3\},\\
    &c_i\in\argmin_{c\in\R\setminus \{\pm \frac{1}{\sqrt{3}}\}} \E_\pi(|\Delta V_\theta(X,c)|\mid X\in\mathcal{A}_i)\qquad\text{for }i\in\{1,2,3\}. \nonumber
    \end{align}

    \item A \textit{Random splitting} where $c_1\sim\mathcal{U}(x_{min},-\frac{1}{\sqrt{3}})$, $c_2\sim\mathcal{U}(-\frac{1}{\sqrt{3}},\frac{1}{\sqrt{3}})$ and $c_3\sim\mathcal{U}(\frac{1}{\sqrt{3}},x_{max})$, where $x_{min}$ and $x_{max}$ are the minimum and maximum values in the sample. 
    
\end{itemize}

\ny
\underline{Local schemes} where the centre of linearization is updated in every step. 

\ny
\begin{itemize}
    \item \textit{Local Strang linearization} where $c_k=x_{t_{k-1}}$ for $k=1,\dots,N.$
    \item \textit{Minimizer of Local Bias }(MoLB) and \textit{Minimizer of Local Variance} (MoLV) defined, respectively, as
    \begin{align}\label{eq: MoLB}
        &c_{t_k}\in\argmin_{c\in\R\setminus \{\pm \frac{1}{\sqrt{3}}\}} |B_\theta(x_{t_{k}},c)|\qquad \text{for }k=1,\dots,N  ,\\
        &c_{t_k}\in\argmin_{c\in\R\setminus \{\pm \frac{1}{\sqrt{3}}\}} |\Delta V_\theta(x_{t_{k}},c)| \qquad \text{for }k=1,\dots,N. \nonumber
    \end{align}
\end{itemize}

\ny
If multiple minima exist for one of the methods in \eqref{eq: MoAB} or \eqref{eq: MoLB}, we select the value among them that minimizes the other error measure. For example, in Figure \ref{figure: bias-variance tradeoff} there are two values of $c$ with $|\Delta V_\theta(x,c)|=0$ for each $x$, and when employing the estimator MoLV, we would then choose the one with the lowest value of $|B_\theta(x,c)|$, which is also highlighted on the figure.

\paragraph{Simulation studies}
We compare the $8$ Strang estimators presented above in $M=1000$ Monte Carlo repetitions with trajectories sampled from the EM scheme using a small time step ($h^{sim}=0.0001$), and sub-sample to analyze different larger time steps. We generate estimates using the  \textit{Nelder-Mead} optimizer from the base \textsf{R} function \texttt{Optim} with starting parameters $(\e_0,y_0,\sigma_0)=(0.5,0,0.1)$ and default stopping criteria. Computations related to objective \eqref{eq: adaptive splitting likelihood} are done in \texttt{Rcpp} for speed-ups, in particular the nonlinear flow function $f_h$ and its derivative are computed with manually implemented Runge-Kutta steps (RK4) and simple central differences, respectively.

\ny
We carry out two simulation studies with the method described above with different parameter values in each experiment. These parameter values are shown in Table \ref{table: parameter values}.

\begin{table}[h!]
\centering
\begin{tabular}{ |c|c|c| } 
\hline
 & Experiment 1 & Experiment 2\\
 \hline
$\e$        & 0.1   & 0.25  \\ 
$y$         & -0.05 & 0  \\ 
$\sigma$  & 1.7   & 1.2   \\ 
 \hline
\end{tabular}
\caption{\textbf{The double-well potential model.} True parameter values used in the simulation studies.}
\label{table: parameter values}
\end{table}

\ny
The results are presented in Figure \ref{fig: double-well strang_estimators}. Comparison of the parameter estimates obtained using the different methods (Figure \ref{fig: double-well estimates}) indicates that the fixed point linearization approach provides accurate estimates, particularly for $y$, and, in Experiment 2, also for $\sigma$. In contrast, $\e$ appears to be more challenging to estimate accurately.

Among the estimators based on minimizing \eqref{eq: bias function} and \eqref{eq: variance function}, MoAV slightly outperforms MoAB. They are both noticeably more accurate than the local splittings MoLB and MoLV, particularly in terms of the variance of the estimates of $y$.

\begin{figure}[h!]
    \centering

    \begin{subfigure}{\linewidth}
        \centering
        \includegraphics[width=\linewidth]{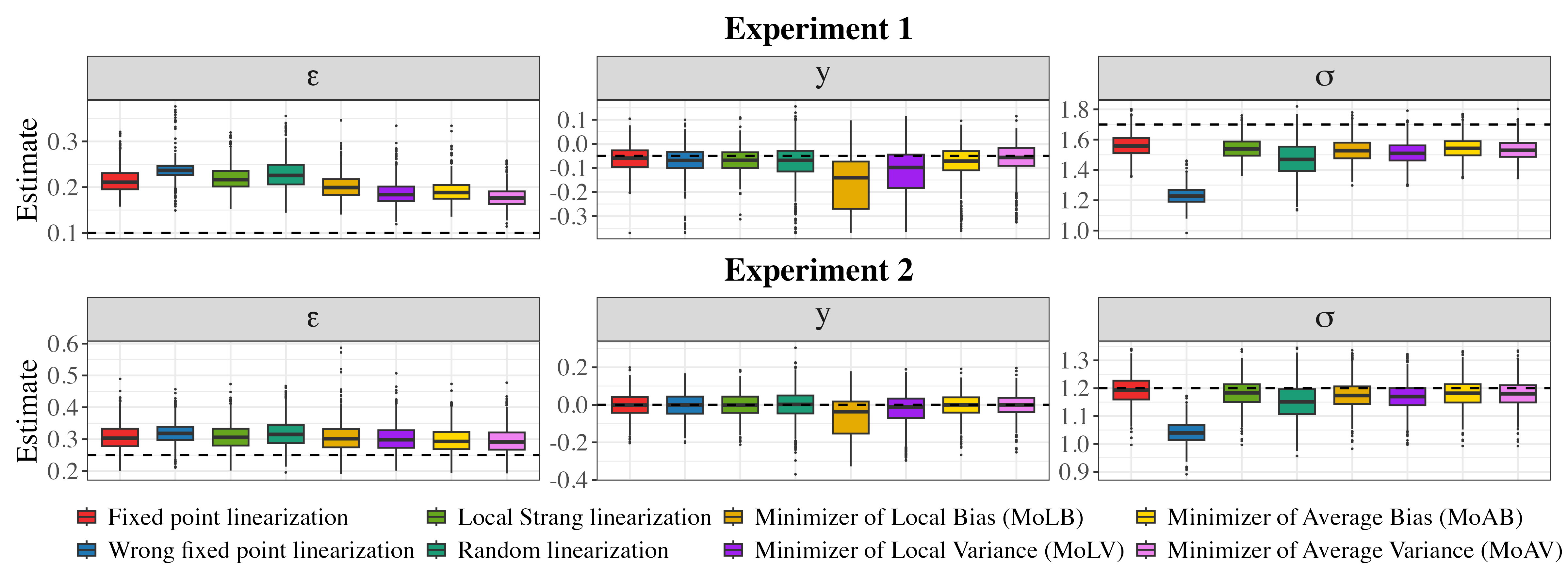}
        \caption{Parameter estimates of eight Strang splitting estimators with step-size $h=0.08$. Black dashed lines are true parameters.}
        \label{fig: double-well estimates}
    \end{subfigure}

    \vspace{0.5cm}

    \begin{subfigure}{\linewidth}
        \centering
        \includegraphics[width=\linewidth]{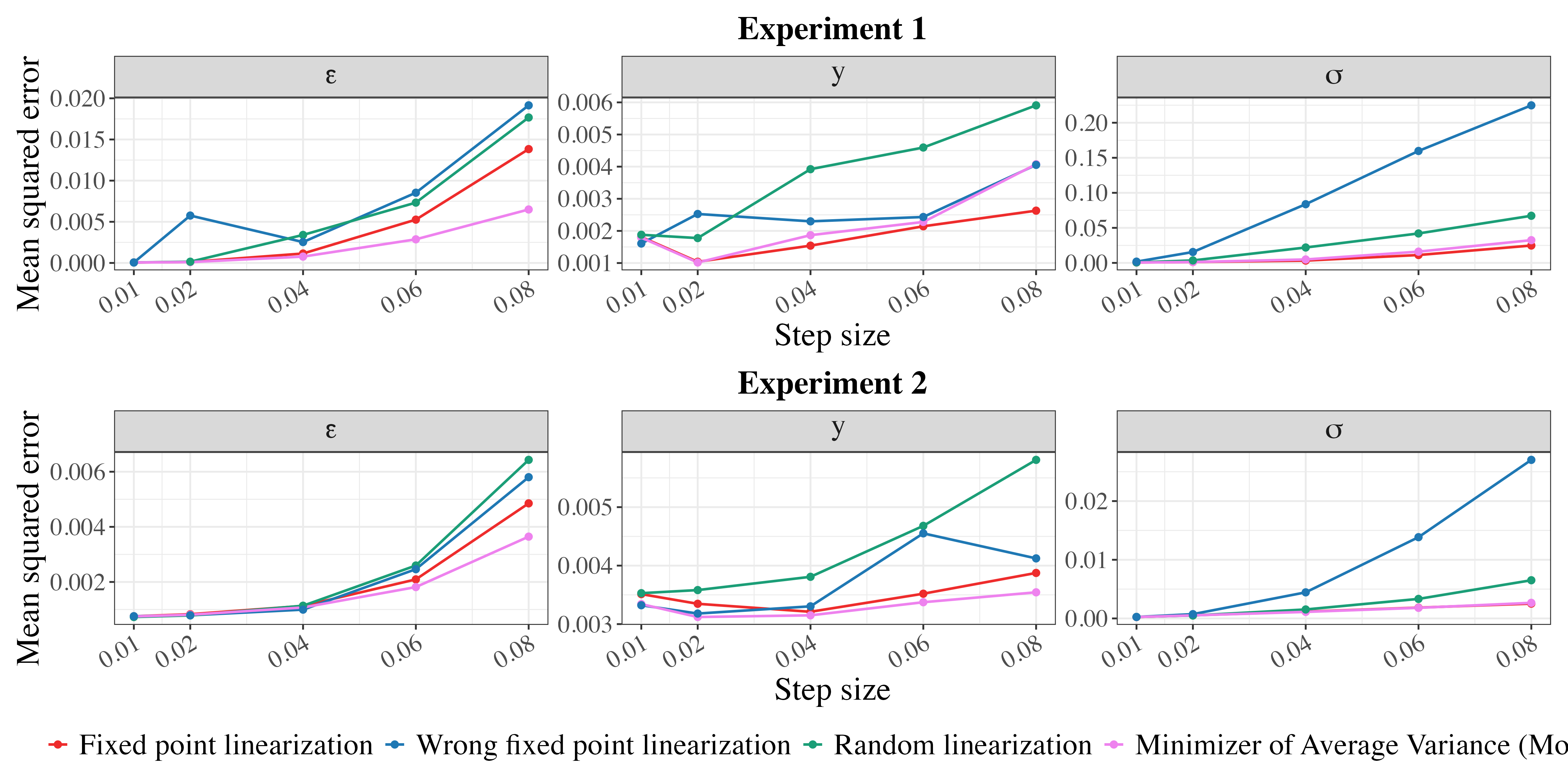}
        \caption{Mean squared error of four Strang splitting estimators for varying step sizes $h\in\{0.01,0.02,0.04,0.06,0.08\}$.}
        \label{fig: double-well mse}
    \end{subfigure}

    \caption{\textbf{Performance of Strang splitting estimators in the double-well potential model.}}
    \label{fig: double-well strang_estimators}
\end{figure}

As expected, the wrong fixed point splitting performs poorly, especially regarding estimates of $\sigma$, but also of $\e$. The random splitting estimator is also quite inaccurate, and these results show that the choice of splitting indeed makes a great impact on the accuracy of the Strang splitting estimator. Surprisingly, even these two methods yield accurate estimates of $y$, particularly in Experiment 2.

The difference in accuracy can be seen even more clearly in Figure \ref{fig: double-well mse}, which shows mean squared error (MSE)  for the estimators of each parameter compared between selected methods with varying step sizes $h$. The figure highlights the two best and the two worst estimators, respectively, and we see that in each case, either the fixed point linearization or MoAV yields the lowest MSE. In particular, the wrong fixed point splitting is sensitive to large step sizes, the error increasing especially fast for the estimates of $\sigma$.

\ny
One additional advantage of the fixed point linearization is that $\sigma$ can be estimated with the analytical expression
\begin{align}\label{eq: sigma estimate}
    \widehat{\sigma}^2 &=
    \frac{2A}{N\left( e^{2Ah}-1\right)} \sum_{k=1}^N  z_{t_k}^2,
\end{align}
the value of $\sigma^2$ that maximizes \eqref{eq:loglikelihood} for a given splitting. This improves the estimation of all parameters, as numerical optimization is performed in a lower-dimensional parameter space. However, the linearization methods that minimize the error measures $|B_\theta(x,c)|$ and $|\Delta V_\theta(x,c)|$ choose $c$ depending on $\sigma$, and the formula \eqref{eq: sigma estimate} depends on $c$ so it cannot be used with these methods. To facilitate a direct and unbiased comparison among the different splitting approaches, we numerically estimated $\sigma^2$ for all methods in the simulation studies, rather than employing the estimator in \eqref{eq: sigma estimate}.

\subsection{A slow--fast excitable system: the FitzHugh-Nagumo model}

The stochastic FitzHugh-Nagumo model is a two-dimensional extension of the double-well potential model given by the SDE
\begin{align*}
    dX_t &= \frac{1}{\varepsilon}(X_t-X_t^3-Y_t+\alpha)\, dt + \sigma_1 \, dW^{(1)}_t\\
    dY_t &= \gamma X_t+\beta-Y_t \, dt + \sigma_2 \, dW^{(2)}_t
\end{align*}
where $0<\varepsilon\ll 1$ is a time scale separation parameter, 
$\alpha ,\beta,\gamma\in\R$ and $\sigma_1,\sigma_2 >0$ are parameters and $W=(W^{(1)},W^{(2)})$ is a two-dimensional Brownian motion. We are particularly interested in parameter values for which a unique fixed point exists and is stable.

\begin{figure}[h]
\begin{subfigure}[b]{0.5\textwidth}
    \centering
    \includegraphics[width=\textwidth]{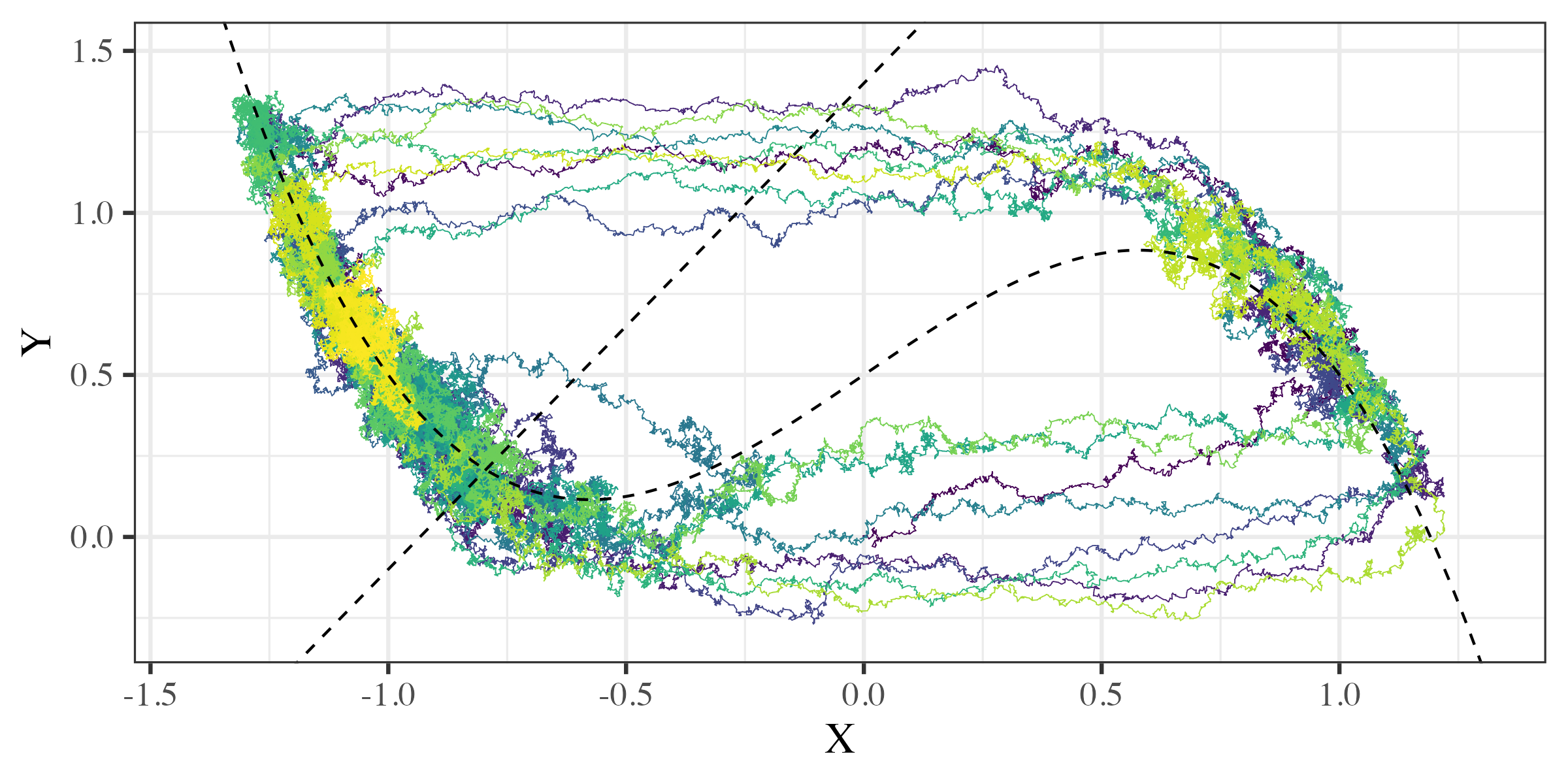}
    \caption{$\beta=1.4$, stable fixed point}
    \label{FHN sample path: fixed point is stable}
\end{subfigure}
\hfill
\begin{subfigure}[b]{0.5\textwidth}
    \centering
    \includegraphics[width=\textwidth]{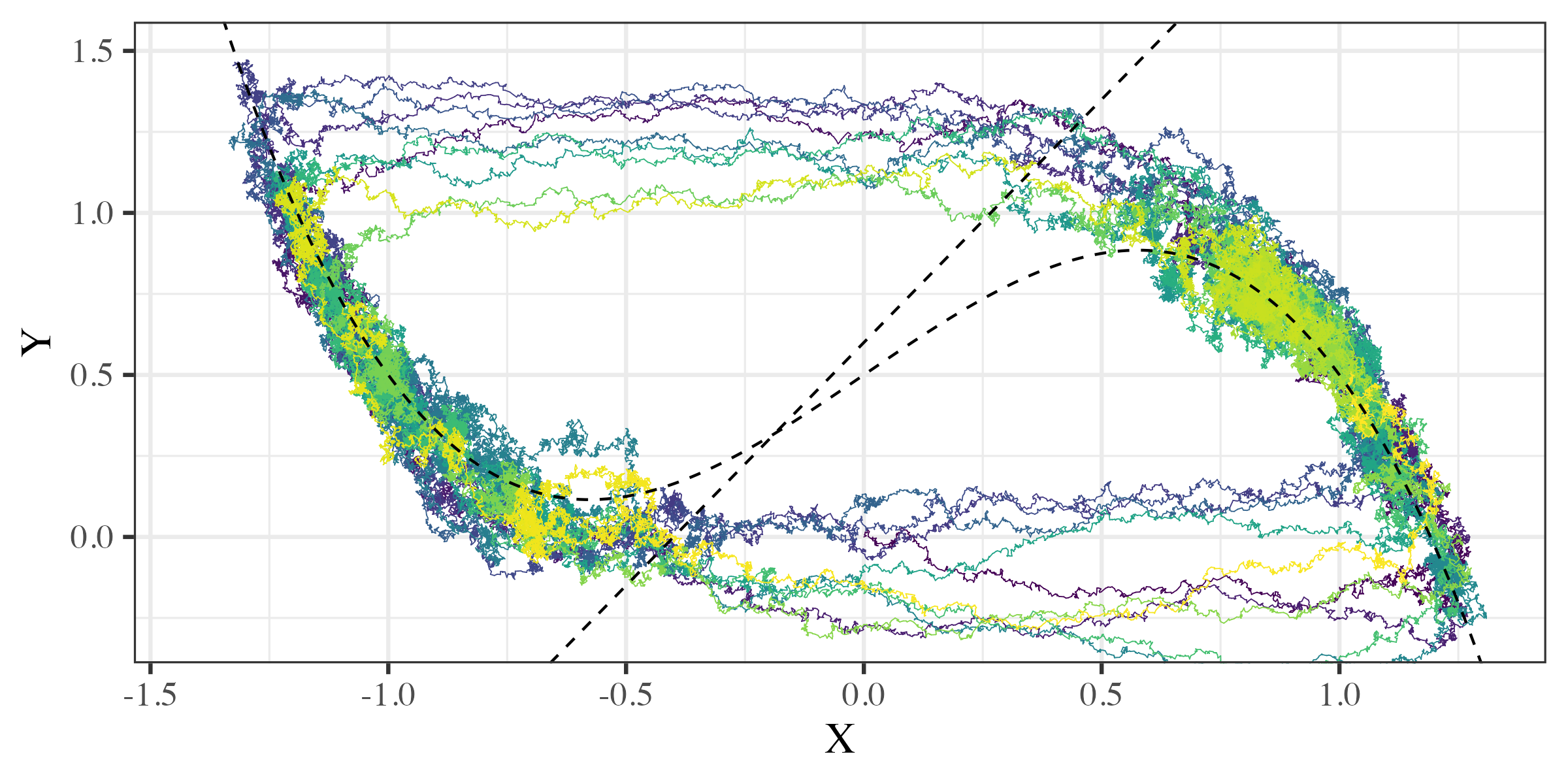}
    \caption{$\beta=0.6$, unstable fixed point}
    \label{FHN sample path: fixed point is unstable}
\end{subfigure}
    \caption{\textbf{Sample paths of the stochastic FitzHugh-Nagumo model}. Simulated from the EM scheme with $h=0.0001$, $T=20$ and parameters $\e=0.05$, $\alpha=0.5,\gamma=1.5$, $\sigma_1=0.3$, $\sigma_2=0.5$ and varying $\beta$. Dashed curves are nullclines.}
    \label{FHN model figure sample paths}
\end{figure}

\ny
A unique fixed point $(x^*,y^*)$ exists for $\gamma>1$ \citep{MCMC} and is then given by
\begin{align*}
    x^* &= \sqrt[3]{U_1}+\sqrt[3]{U_2},
    \quad 
    y^* = \gamma x^*+\beta,
\end{align*}
where $U_1=\frac{\alpha-\beta}{2}+\sqrt{D}, U_2 = \frac{\alpha-\beta}{2}-\sqrt{D}$ and $D=\frac{(\beta-\alpha)^2}{4}+\frac{(\gamma-1)^3}{27}$.
As the eigenvalues of a $2\times 2$ matrix $A$ have negative real parts if and only if $\det(A)>0$ and $\Tr(A)<0$ \citep{dynamics}, the fixed point is stable for
\begin{align*}
    \Tr\left( DF(x^*,y^*) \right) &= \frac{1}{\e}\left(  1-\e-3(x^*)^2   \right) <0,\\
    \det\left( DF(x^*,y^*) \right) &= \frac{1}{\e}\left(3(x^*)^2-1+\gamma\right) >0.
\end{align*}
The determinant is always positive when $\gamma>1$ so the fixed point is stable if and only if $|x^*|>\sqrt{\frac{1-\e}{3}}$. Figure \ref{FHN model figure sample paths} shows two sample paths corresponding to different parameter values, for which the fixed point, given by the intersection of the two nullclines, is respectively stable and unstable. Due to stochastic perturbations, the FitzHugh-Nagumo model is an excitable system making large excursions in state space even when the fixed point is stable. Such excursions occur when the process is perturbed below the left knee of the nullcline.

\ny
The splitting scheme with linearization around  $c=(c_1,c_2)\in\R^2$ with $b=c-A^{-1}F(c)$ is given by
\begin{align}
A
    &=
    \begin{pmatrix}
    \frac{1-3c_1^2}{\varepsilon}&-\frac{1}{\varepsilon}\\
    \gamma & -1
    \end{pmatrix}\nonumber\\
A\left(\begin{pmatrix}
        X_t\\ Y_t
     \end{pmatrix} -b\right)
     &= \begin{pmatrix}
        \frac{1}{\e}\left(2c_1^3-3c_1^2X_t+X_t-Y_t+\alpha\right) \\
        \gamma X_t - Y_t + \beta
      \end{pmatrix}
    \label{eq: FHN splitting: linear part A(x-b)}
    \\
N(X_t,Y_t) &= 
    \begin{pmatrix}
    \frac{1}{\e}(-X_t^3+3c_1^2X_t-2c_1^3)\\
    0
    \end{pmatrix} .
    \label{eq: FHN splitting: N(x,y)}
\end{align}
\noindent
When $\gamma>1,$ $A$ is invertible for all $c.$ The splitting does not depend on $c_2$, which is a consequence of $Y_t$ entering only linearly in $F$.

\ny
As the FitzHugh-Nagumo system is a two-dimensional model, the error functions from Theorem \ref{nyt teorem} are not available.
Instead, we visualize the accuracy of Strang one-step predictions in other ways.

\begin{figure}[h]
    \centering
    \includegraphics[width=\linewidth]{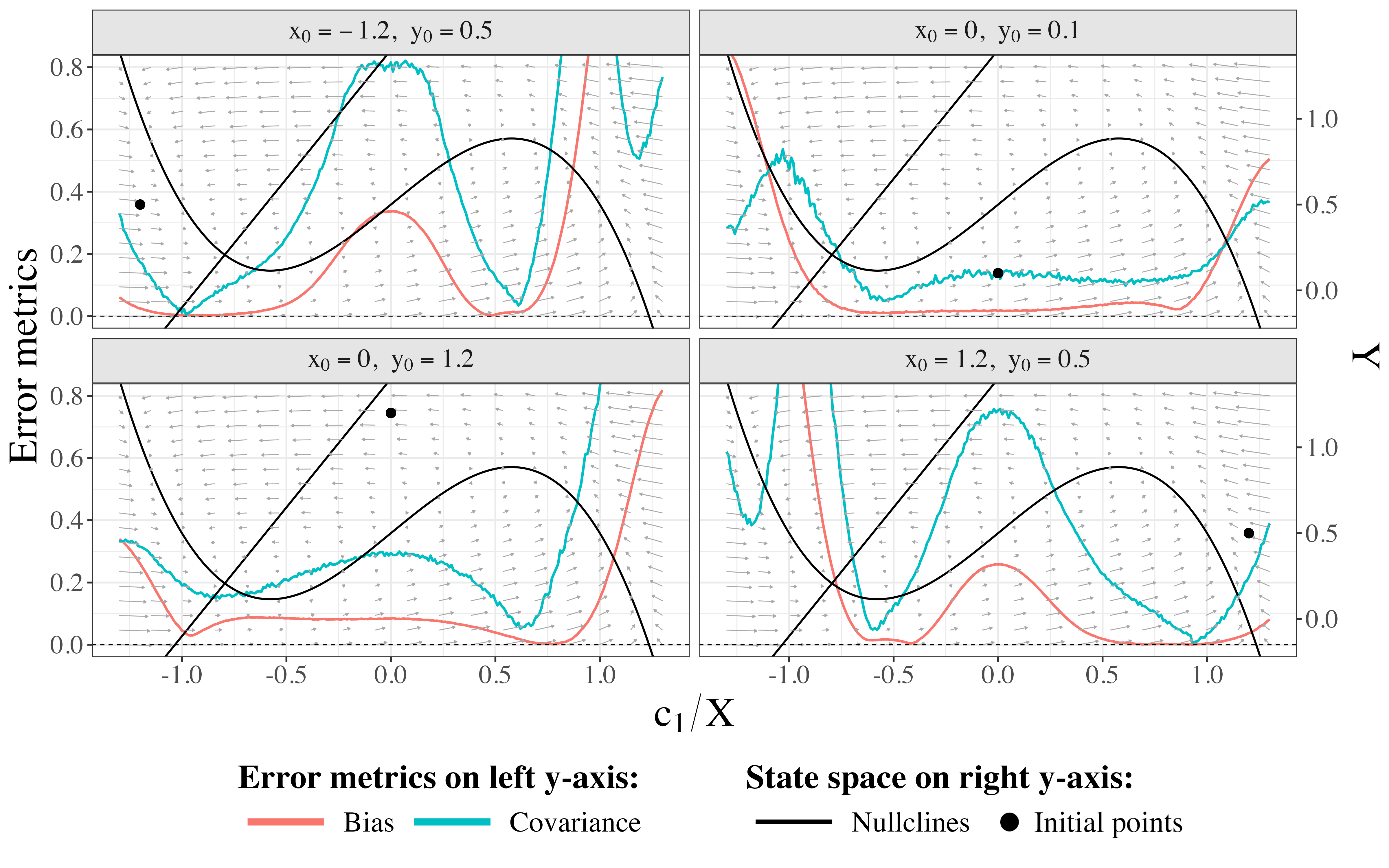}
    \caption{\textbf{Empirical errors of Strang steps of size $\mathbf{h=0.08}$ for varying centres of linearization $\mathbf{c=(c_1,c_2)}$ for the FitzHugh-Nagumo model.}
    The curves show empirical bias $\E(||X^{[S]}_{t}-X_t||)$ and relative deviation of the empirical covariance matrices $||\Cov(X^{[S]}_t)-\Cov(X_t)||_F/||\Cov(X_t)||_F$ conditional on $(x_0,y_0)$.
    Each plot uses a different initial value.}
    \label{plot, FHN: bias and covariance}
\end{figure}

Since the splitting scheme linearizing around $c$ does not depend on $c_2$, we only need to consider values of $c_1$ to compare splittings.
Given $c_1\in\R$ and an initial value $(x_0,y_0)\in\R^2$, we simulate $5000$ iid. Strang one-step predictions for a corresponding linearization with step size $h=0.08$.
We also simulated trajectories using the EM scheme with a substantially smaller step size, $h^{sim}=0.0001$, rendering the discretization error negligible. These simulations provide an approximation to the true distribution at time step $h$. We then compared their empirical mean and variance with those of the Strang one-step prediction.

The results are shown in Figure \ref{plot, FHN: bias and covariance}, where we plot empirically computed values of bias $\E(||X^{[S]}_{t}-X_t||)$ as well as deviation of covariance $||\Cov(X^{[S]}_t)-\Cov(X_t)||_F/||\Cov(X_t)||_F$, where  $||\cdot||_F$ is the Frobenius norm, for continuously varying $c_1$.

There are major differences in bias and covariances for different linearizations, yet it is not easy to determine a pattern. 
One important observation is that both the bias and the deviation of covariance are extremely large when using the fixed point linearization for the initial value $(x_0,y_0)=(1.2, 0.5)$.
Also, Figure \ref{plot, FHN: bias and covariance} suggests that $\pm \frac{1}{\sqrt{3}}$ are good choices of $c_1$, with the sign depending on the initial value. For $|x_0|=1.2$, the error metrics are also minimized close to either $c_1=1$ or $c_1=-1$.

\ny
We therefore propose the following adaptive splitting, choosing $c_1$ depending on the current value $(x,y)$ of the process:
\begin{align}\label{FHN adaptive splitting}
   (c_1,c_2)= \begin{cases}
(-1,\alpha) & \text{if } x < \frac{-1}{\sqrt{3}},\\
(\frac{-1}{\sqrt{3}},
\alpha)  & \text{if } x \in [\frac{-1}{\sqrt{3}}, \frac{1}{\sqrt{3}}], \ y\leq \alpha ,\\
(\frac{1}{\sqrt{3}},
\alpha)  & \text{if } x \in [\frac{-1}{\sqrt{3}}, \frac{1}{\sqrt{3}}], \ y> \alpha,\\
(1,\alpha)  & \text{if } x > \frac{1}{\sqrt{3}}.
\end{cases} 
\end{align}
With $c_1=\pm \frac{1}{\sqrt{3}}$ we obtain the same matrix $A$ as in \cite{buckwar_etal}, though the splitting proposed in that paper differs by not including the term $b$.

\begin{figure}[h]
    \centering
    \includegraphics[width=\linewidth]{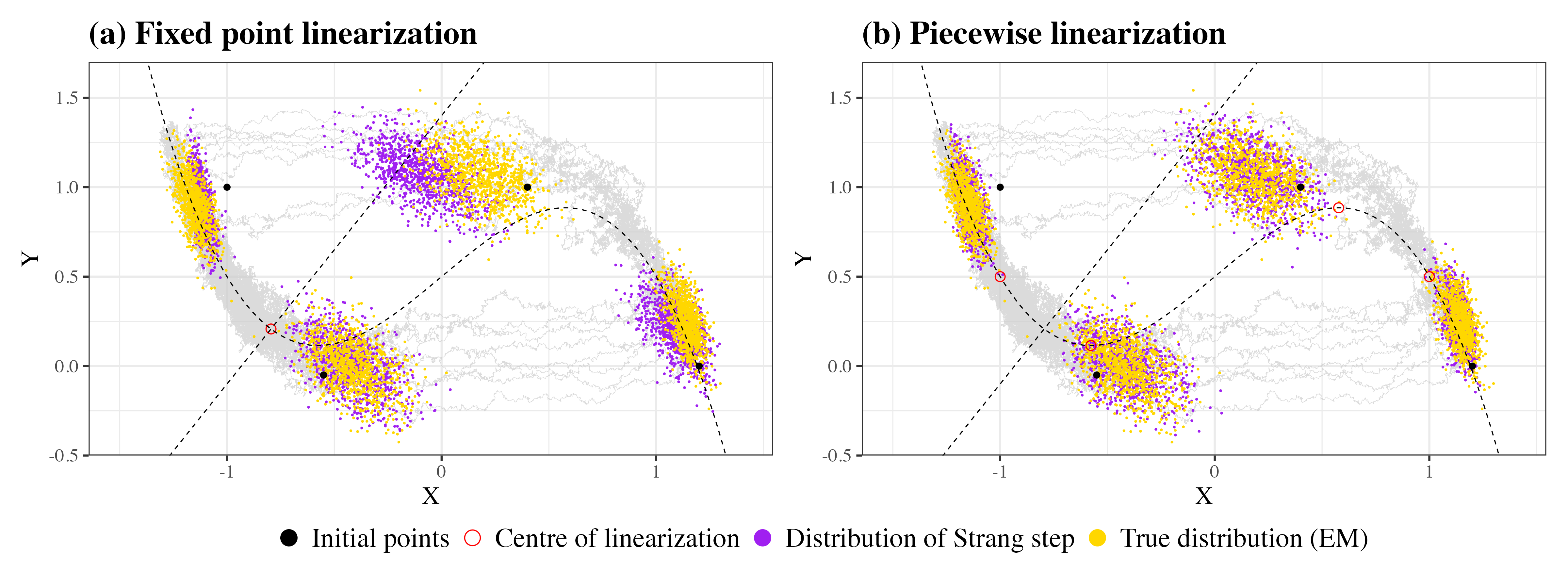}
    \caption{\textbf{Strang one-step predictions compared to true distributions for various starting points $\mathbf{(x_0,y_0)}$ for the FitzHugh-Nagumo model.}  The step size is $h=0.08$. True distributions are approximated by accumulated EM-steps of size $h^{sim}=0.0001$. Example trajectories (grey curves) and nullclines (dashed grey lines) are overlaid for reference.}
    \label{skyer: forskellige Strang one-predictions}
\end{figure}

\ny
To further investigate this, we plot clusters of simulated values from Strang splitting schemes and the true distribution, respectively, simulating one-step predictions as in Figure \ref{plot, FHN: bias and covariance}.
Figure \ref{skyer: forskellige Strang one-predictions} shows such clusters of 1000 simulated values for four different initial values.
When a cluster of Strang one-step predictions matches up well with a cluster of values from the true distribution, we conclude that the Strang one-step prediction is an accurate approximation of the transition density.

This figure verifies that the choice of splitting is an important factor for the accuracy of the Strang splitting scheme.
In particular, Figure \ref{skyer: forskellige Strang one-predictions}(a) shows that the fixed point linearization scheme performs well in large parts of state space but not everywhere.
On the other hand, the piecewise linearization scheme \eqref{FHN adaptive splitting} used in Figure \ref{skyer: forskellige Strang one-predictions}(b) yields accurate distributions for all four initial values.

\ny
In the terminology of slow--fast systems, the nullcline of $X$ is divided into three slow manifolds \citep[Definition 2.1.1]{berglund}. The two sections of the nullcline with $|x|>\frac{1}{\sqrt{3}}$ are stable slow manifolds while the section in between is unstable. An intuitive understanding of \eqref{FHN adaptive splitting} is that linearizing around a point on a stable slow manifold is accurate when the process is currently near the same manifold, and the points $(\pm 1,\alpha)$ lie on the stable slow manifolds for all $\theta$.

For a general slow-fast model, we thus propose to use an adaptive splitting with a linearization corresponding to each stable slow manifold. This implies that one must first identify the slow manifolds of the system and their stability and then specify a way of classifying whether the process is currently near one of the stable manifolds. For the FitzHugh-Nagumo system, this is easy since we only have to check the value of $X_t$ to determine if the process is currently moving along a stable slow manifold.

\paragraph{Simulation study}
We implement the Strang splitting estimator for the FitzHugh-Nagumo model and estimate parameters in two separate simulation studies with both a stable and an unstable fixed point, respectively (see Figure \ref{FHN model figure sample paths}), using the same computational methods as in Section \ref{subsection double-well}.
The true parameters in the case of a stable fixed point are $\e= 0.1$, $\alpha=0.5$, $\gamma=1.5$, $\beta=1.4$, $\sigma_1 = 0.3$ and $\sigma_2=0.5$, which results in the model being excitable. In the second simulation study, we instead use $\beta=0.6$, yielding an unstable fixed point and oscillatory behaviour of the system.

The results are shown in Figure \ref{fig: FHN estimates}.
For the excitatory model, the figure shows that in particular estimates of $\e$ and $\sigma_1$ are less biased with the piecewise linearization than with the fixed point linearization. The figure also shows MSE for each parameter for $h\in\{0.02, 0.04, 0.06 , 0.08\}$.

For the oscillatory model, the difference between the two estimators is even more pronounced, which is to be expected since the process in this case spends more time in those regions of state space where the fixed point linearization is inaccurate.
Figure \ref{fig: FHN estimates} shows that for the oscillatory model, the piecewise linearization yields a Strang splitting estimator that is more accurate for all six parameters, with the difference in performance between this scheme and the fixed point linearization growing as $h$ increases.
The piecewise linearization is in particular more accurate when estimating $\e$, $\gamma$ and $\sigma_1$ and the estimates of $\beta$ have a much smaller variance.

\begin{figure}[h!]
    \centering

    \begin{subfigure}{\linewidth}
        \centering
        \includegraphics[width=\linewidth]{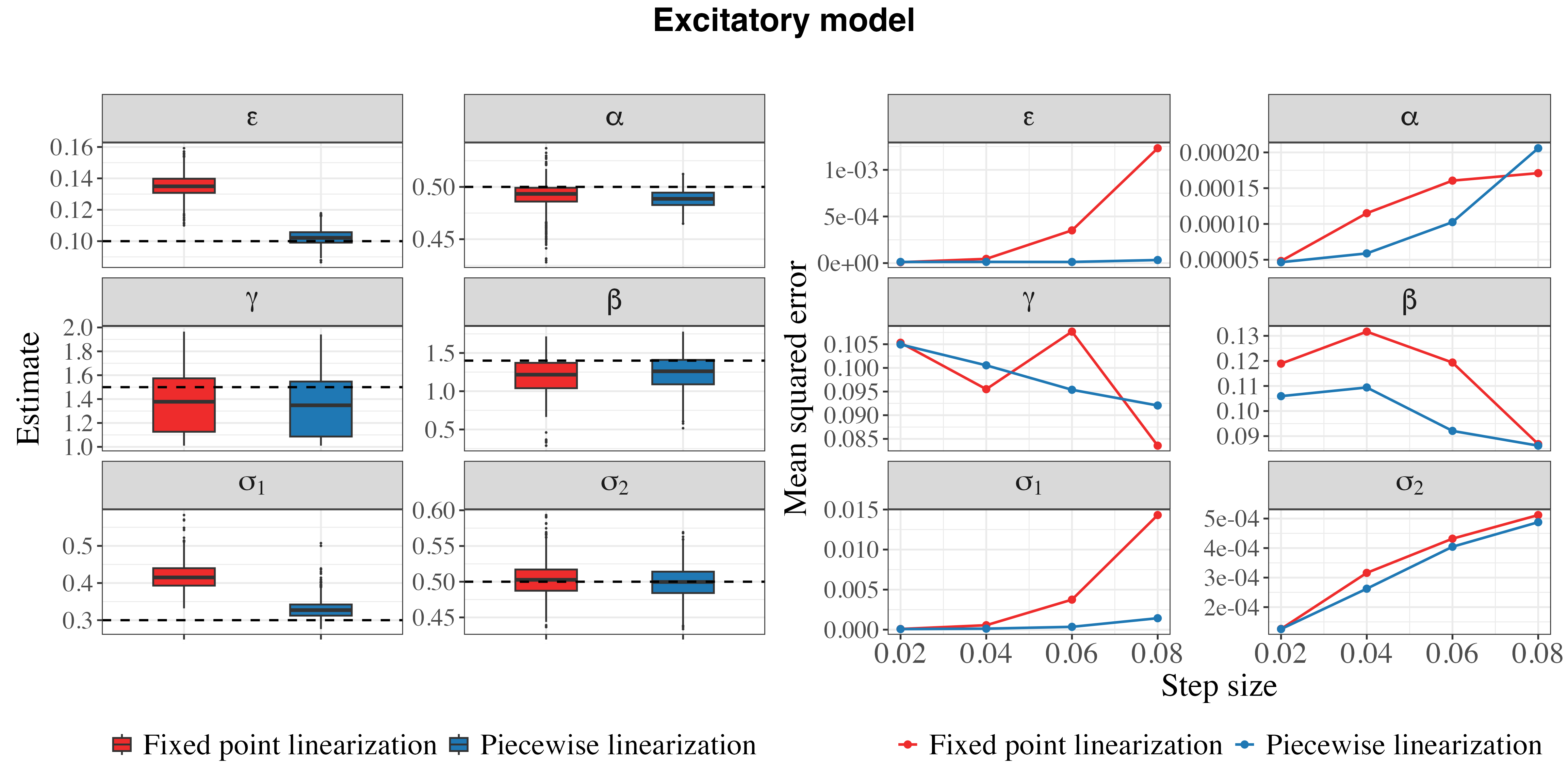}

        \label{fig: Excitatory estimates}
    \end{subfigure}

    \vspace{0.5cm}

    \begin{subfigure}{\linewidth}
        \centering
        \includegraphics[width=\linewidth]{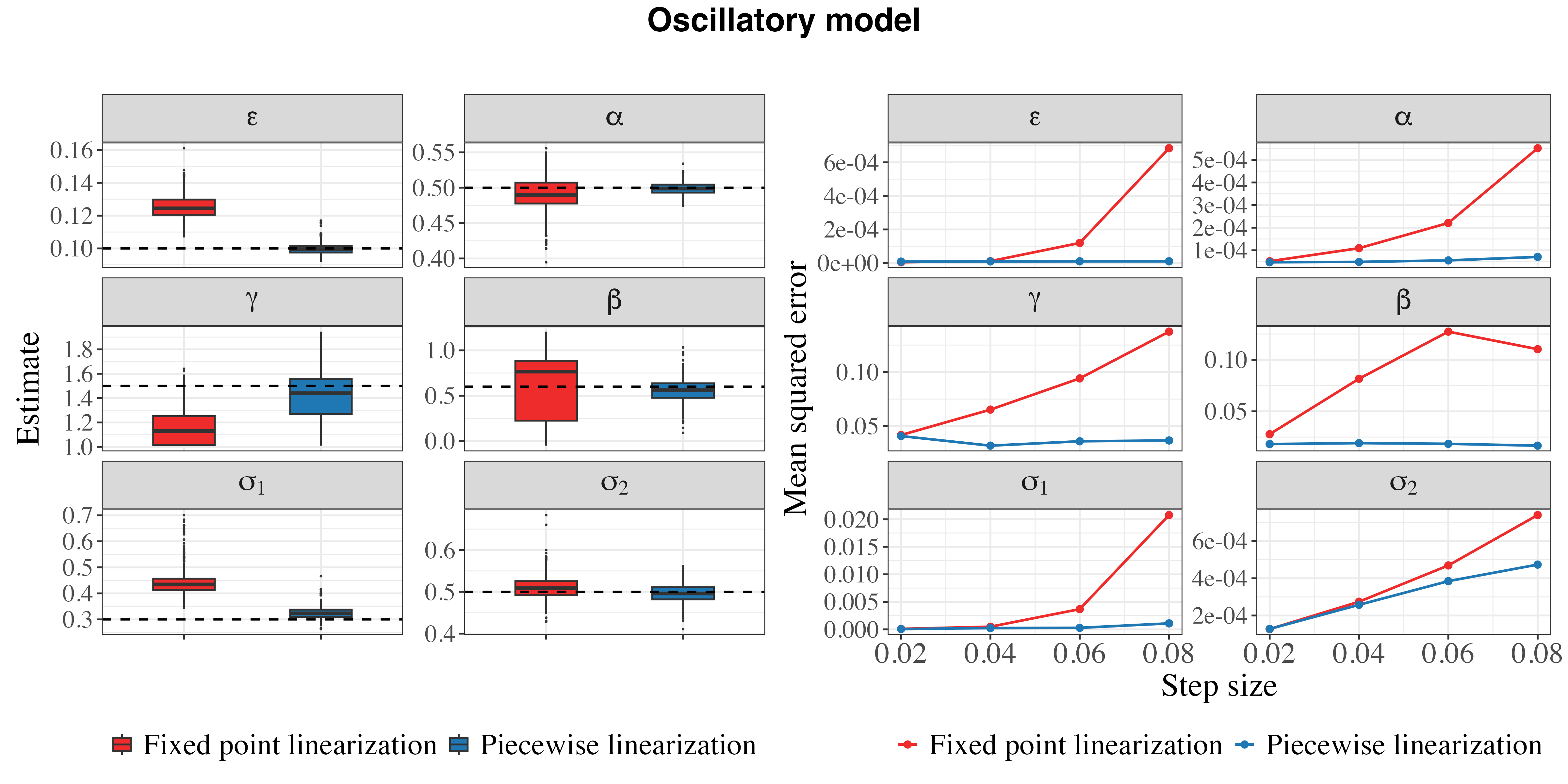}
     
        \label{fig: Oscillatory estimates}
    \end{subfigure}

    \caption{\textbf{Performance of Strang splitting estimators in the FitzHugh-Nagumo model.} Boxplots and mean squared errors are based on 1000 simulated data sets. Black dashed lines are true parameters.}
    \label{fig: FHN estimates}
\end{figure}

\section{Conclusion}

The Strang splitting estimator provides a computationally attractive framework for likelihood-based inference in nonlinear SDE models, but its practical performance depends critically on the chosen decomposition of the drift into linear and nonlinear components. In this paper, we have investigated this dependence both theoretically and empirically. By deriving higher-order expansions for the moments of the Strang one-step prediction in Theorem \ref{nyt teorem}, we showed that the approximation errors in both the conditional expectation and the conditional variance are of order $h^3$, and that the leading error terms depend explicitly on the selected splitting. These results provide a theoretical explanation for the substantial finite-sample differences previously observed between different Strang splittings.

For one-dimensional models, the error expansions lead naturally to the bias and variance functions $B_\theta$ and $\Delta V_\theta$, which can be used to assess the quality of a splitting independently of parameter estimation. In the double-well potential model, these quantities reveal that linearization around stable fixed points yields transition densities with small local errors over the regions of state space where the process spends most of its time. Consistent with this observation, the corresponding fixed point linearization scheme produced accurate parameter estimates and low mean squared errors across all simulation settings. Although splittings obtained by directly minimizing the derived error measures can occasionally produce marginal improvements, the gains were generally modest compared with the additional computational complexity required.

For ergodic models, the long-term averages of the bias and the error in the covariance across an entire sample path can be studied. For the double-well potential model, the mean bias is approximately minimized by choosing the Strang splitting adaptively as a linearization around the fixed points.
Along with simulation studies showing that we obtain accurate parameter estimates with this splitting scheme, this lends support to the use of the fixed point linearization scheme for potential models. 

Beyond its favorable approximation properties, the fixed point linearization scheme offers a considerable practical advantage through the existence of an explicit estimator of $\Sigma$. Since the pseudo-likelihood can be maximized analytically with respect to the diffusion coefficient for a fixed splitting, estimation can be performed in a lower-dimensional parameter space. This not only reduces computational cost but also tends to improve the robustness of the optimization procedure by mitigating identifiability issues and decreasing the risk of convergence to suboptimal local maxima. Consequently, the benefits of fixed point linearization arise both from more accurate transition-density approximations and from a simpler and more stable estimation problem.

For the stochastic FitzHugh-Nagumo system, we obtain better performance of the Strang splitting estimator by choosing an adaptive splitting depending on local dynamics rather than using the fixed point linearization scheme, even though the model has a unique stable fixed point. The FitzHugh-Nagumo model demonstrates that fixed point linearization is not universally optimal. In excitable slow-fast systems, trajectories frequently undergo large excursions away from equilibrium, so local dynamics near the stable fixed point are not representative of the regions most relevant to inference. For this model, a piecewise adaptive splitting based on the geometry of the slow manifolds yielded substantially more accurate approximations of the transition density and improved parameter estimation. This suggests that the most effective splitting should reflect the dominant dynamical structures of the system rather than solely its equilibria. We conjecture that similar results may hold in the wider class of excitable stochastic systems.

Taken together, our findings indicate that optimal splitting strategies are model dependent. For gradient and potential-type systems, fixed point linearization appears to be a robust and practically useful default choice, while for slow-fast excitable systems adaptive splittings guided by the underlying phase-space geometry may be preferable. 

Several directions for future work remain open. 
First, extending the moment-expansion results to multidimensional systems would provide a more systematic basis for selecting splittings in higher dimensions. Such an extension would require a more elaborate mathematical framework, as higher-order derivatives and expansions are most naturally expressed using multi-index or tensor notation, leading to substantially more complex expressions for the corresponding bias and covariance error terms. Second, we conjecture that the favorable performance of fixed point linearization observed for the univariate potential model extends straightforwardly to multivariate gradient systems, since in both settings the dynamics are predominantly governed by local behavior near stable minima of the potential, where the linear approximation is most accurate. Third, adaptive splittings based on slow manifolds, invariant manifolds, or other geometric structures deserve further investigation, particularly for excitable, oscillatory, and chaotic systems. Finally, it remains an open question whether a general notion of optimal splitting can be formulated in terms of information-theoretic criteria, such as minimizing the discrepancy between the true and approximate transition densities.

Overall, the results demonstrate that the choice of splitting is not merely a numerical detail but a central component of the inference procedure. Careful exploitation of model structure can substantially improve the accuracy of Strang splitting estimators and broaden their applicability to a wider class of nonlinear stochastic dynamical systems.

\section{Proofs}\label{section: proofs}

\ny
The proof of Theorem \ref{nyt teorem} relies on expanding the moments of $\Phi^{[S]}_h(X_{t_k})$. For this purpose we extend the approximation of the flow function $f_h$ \citep[Proposition 2.2]{parameterestimation} up to order $h^3.$ 
\begin{lemma}\label{lemma udvidelse af f}
    Assume that $f_h$ is $C^4.$ Then for all $x\in\R,$
    \begin{align*}
        f_h(x)=x+hN(x)+\frac{h^2}{2}N'(x)N(x) 
        +\frac{h^3}{6} \left( N''(x)N(x)^2+N'(x)^2N(x)\right)
        +O(h^4).
    \end{align*}
\end{lemma}
\begin{proof}
Assume $X=(X_t)_{t\geq 0}$ solves the ODE $X_t'=N(X_t)$ and let $t\geq0$ be given.
We Taylor expand $h \mapsto X_{t+h}$ around 0 and get
\begin{align*}
    f_h(X_t) 
    &= f_0(X_t) + h \frac{d}{d h}f_h(X_t)\Big|_{h=0}
    + \frac{h^2}{2} \frac{d^2}{dh^2}f_h(X_t)\Big|_{h=0}
    + \frac{h^3}{6} \frac{d^3}{dh^3}f_h(X_t)\Big|_{h=0}
    +O(h^4)\\
&= X_t + hN(X_t) + \frac{h^2}{2}N'(X_t)N(X_t)
    + \frac{h^3}{6}
    \Big( N''(X_t)N(X_t)^2+N'(X_t)^2N(X_t) \Big)
    + O(h^4).
\end{align*}
\end{proof}

\ny
\begin{proof}[Proof of Theorem \ref{nyt teorem}]
    The proof is similar to that of Proposition 3.6 in \cite{parameterestimation}.
    \\
(i)
The Taylor expansion of $\mu_h (y) = e^{Ah}y + \left(1-e^{Ah}\right)b$ around $h=0$ is
\begin{align*}
    \mu_h\left( y \right) 
    &= y + hA\left(y-b\right)
    + \frac{h^2}{2}A^2\left(y-b\right)
    + \frac{h^3}{6}A^3\left(y-b\right)+ O(h^4)
\end{align*}
so expanding each term up to order $h^3,$ using the expansion of $y=f_{\frac{h}{2}}(x)$ from Lemma \ref{lemma udvidelse af f}, we get
\begin{align}
\label{eq: uh_expansion}
     \mu_h\left( f_{\frac{h}{2}}(x) \right)
     &=
    x + h \left(\frac{1}{2}N(x)+A(x-b)\right)
    + h^2 \left( \frac{1}{8}N'(x)N(x) + \frac{1}{2}AN(x)+\frac{1}{2}A^2(x-b)\right) \nonumber \\
    &\quad + h^3 \left(  \frac{1}{48}\big(N''(x)N(x)^2+N'(x)^2N(x)\big) + \frac{1}{8}AN'(x)N(x) + \frac{1}{4}A^2N(x) + \frac{1}{6}A^3(x-b)    \right) \nonumber \\
    &\quad + O(h^4).
\end{align}
Defining $\widetilde{Q}_h(x)= \mu_h( f_{\frac{h}{2}}(x) )-x$, we have
\begin{align*}
    \E\left[ X^{[S]}_{t_k} \ | \ X^{[S]}_{t_{k-1}}=x \right]
    &= \E\left[ f_\frac{h}{2}\left( \mu_h\big( f_{\frac{h}{2}}(x) \big) + \xi_h\right)  \right]= \E\left[f_{\frac{h}{2}}(x+\widetilde{Q}_h(x)+ \xi_h) \right].
\end{align*}
Lemma \ref{lemma udvidelse af f} then yields
\begin{align}\label{eq: fmuf før expectation}
    \nonumber
    f_{\frac{h}{2}}(x+\widetilde{Q}_h(x) + \xi_h)
    &= x+\widetilde{Q}_h(x)+ \xi_h
    + \frac{h}{2}N(x+\widetilde{Q}_h(x)+ \xi_h)\\
    \nonumber
     &
    +\frac{h^2}{8} N'(x+\widetilde{Q}_h(x)+ \xi_h)N(x+\widetilde{Q}_h(x)+ \xi_h)\\
    \nonumber
    & + \frac{h^3}{48}\Big( N''(x+\widetilde{Q}_h(x)+ \xi_h)N(x+\widetilde{Q}_h(x)+ \xi_h)^2\\
     &\quad\quad\quad+ N'(x+\widetilde{Q}_h(x)+ \xi_h)^2N(x+\widetilde{Q}_h(x)+ \xi_h)\Big)
    + O(h^4).
\end{align}
Before taking the expectation of \eqref{eq: fmuf før expectation},
we first find the order of  $\E[(\widetilde{Q}_h(x)+\xi_h)^n]$ for a given $n.$ Since the odd moments of $\xi_h$ are zero, $\E[\xi_h^2]=h\sigma^2+h^2A\sigma^2+O(h^3)$ (see \eqref{eq: Omega_h}) and $\widetilde{Q}_h(x)= O(h)$ (see \eqref{eq: uh_expansion}), we have
\begin{align*}
    \E\left[(\widetilde{Q}_h(x)+\xi_h)\right] &= \widetilde{Q}_h(x)= O(h),\\
    \E\left[(\widetilde{Q}_h(x)+\xi_h)^2\right] &= \widetilde{Q}_h(x)^2 + h\sigma^2+ h^2A\sigma^2 +O(h^3)= O(h),\\
    \E\left[(\widetilde{Q}_h(x)+\xi_h)^3\right] &= 
    \widetilde{Q}_h(x)^3 +3 \widetilde{Q}_h(x)(h\sigma^2+h^2A\sigma^2)
    =O(h^2),\\
    \E\left[(\widetilde{Q}_h(x)+\xi_h)^4\right] &=
    3h^2\sigma^4 + O(h^3) =   O(h^2).
\end{align*}
For $n\geq 5$, the expected value is $O(h^3).$ We now find the expectation of each term in \eqref{eq: fmuf før expectation}.
Taylor expanding $N$ and using assumption \textbf{A3}
to bound the Lagrange remainder yields
\begin{align*}
    \E\left[\frac{h}{2}N(x+\widetilde{Q}_h(x)+ \xi_h) \right]
    &= \frac{h}{2}\E\Big[
    N(x) + N'(x)(\widetilde{Q}_h(x)+ \xi_h ) 
    + \frac{1}{2}N''(x)(\widetilde{Q}_h(x)+ \xi_h)^2\\
    &\quad + \frac{1}{6}N'''(x)(\widetilde{Q}_h(x)+ \xi_h)^3+\frac{1}{24}N''''(x)(\widetilde{Q}_h(x)+ \xi_h)^4  \Big]
    +O(h^4)\\
     &= \frac{h}{2}\Bigg(N(x) + N'(x)\widetilde{Q}_h(x)
    + \frac{1}{2}N''(x)\left(\widetilde{Q}_h(x)^2+h \sigma^2+h^2A\sigma^2 \right)\\
    &\quad + \frac{1}{2}N'''(x)\widetilde{Q}_h(x)h\sigma^2 + \frac{1}{8}N''''(x)h^2\sigma^4
    \Bigg)       +O(h^4)\\
    &\stackrel{\eqref{eq: uh_expansion}}{=} \frac{h}{2}N(x) + \frac{h^2}{2}N'(x)\left(\frac{1}{2}N(x)+A(x-b)\right) 
    \\& \quad
    + \frac{h^3}{2}N'(x)\left(\frac{1}{8}N'(x)N(x) + \frac{1}{2}AN(x)+\frac{1}{2}A^2(x-b)\right)\\
    &\quad + \frac{h^3}{4}N''(x)\left(\frac{1}{2}N(x)+A(x-b)\right)^2+ \frac{h^2}{4}N''(x)\sigma^2 + \frac{h^3}{4}AN''(x)\sigma^2\\
    &\quad + \frac{h^3}{4}N'''(x)\sigma^2 \left(\frac{1}{2}N(x)+A(x-b)\right)
    + \frac{h^3}{16}N''''(x)\sigma^4
    + O(h^4).
\end{align*}
Similarly, expansions of $N'$ and $N$ yield
\begin{align*}
    &\frac{h^2}{8}\E \left[ N'\big(x+\widetilde{Q}_h(x)+\xi_h\big)N\big(x+\widetilde{Q}_h(x)+\xi_h\big) \right]\\
    &=
    \frac{h^2}{8}
    \E\bigg[  \left(N'(x)+N''(x)\big(\widetilde{Q}_h(x)+\xi_h\big)+\frac{1}{2}N'''(x)\big(\widetilde{Q}_h(x)+\xi_h\big)^2 \right) \\
    &\qquad\qquad \cdot
     \left( N(x)+N'(x)\big(\widetilde{Q}_h(x)+\xi_h\big)+\frac{1}{2}N''(x)\big(\widetilde{Q}_h(x)+\xi_h\big)^2
    \right)            \bigg] + O(h^4)\\
    & =
    \frac{h^2}{8} N'(x)N(x) + \frac{h^3}{8} N''(x)N(x)\left(\frac{1}{2}N(x)+A(x-b)\right)    
    + \frac{h^3}{8}N'(x)^2\left(\frac{1}{2}N(x)+A(x-b)\right)\\
    &\quad 
     + \frac{h^3}{16}N'''(x)N(x) \sigma^2 
    + \frac{3h^3}{16}N''(x)N'(x) \sigma^2 + O(h^4).
\end{align*}
Finally, we have
\begin{align*}
    &\frac{h^3}{48} \E\left[\left( N''\big(  x+\widetilde{Q}_h(x)+ \xi_h\big)N\big(x+\widetilde{Q}_h(x)+ \xi_h\big)^2+ N'\big(x+\widetilde{Q}_h(x)+ \xi_h\big)^2N\big(x+\widetilde{Q}_h(x)+ \xi_h\big)\right)\right]\\
    &= \frac{h^3}{48} \left(
    N''(x)N(x)^2 + N'(x)^2N(x) 
    \right)+ O(h^4)
\end{align*}
since all additional terms in expansions of $N, N'$ or $N''$ are at least of order $h$. Combining all yields
\begin{align}\label{eq: Strang expectation}
    \nonumber
    \E\Big[  f_{\frac{h}{2}}&(x+\widetilde{Q}_h(x)+ \xi_h) \Big]=
    x + hF(x) 
    +  \frac{h^2}{2}\left(F'(x)F(x)+\frac{1}{2}\sigma^2 F''(x)\right)\\
    \nonumber
    \quad 
    &+ h^3 \Bigg(
    \frac{3}{8}AN'(x)N(x) + \frac{1}{4}A^2N(x) + \frac{1}{6}A^3(x-b) + \frac{1}{6}N''(x)N(x)^2 +\frac{1}{6}N'(x)^2N(x)
    \\
    \nonumber
    &\quad  
      +\frac{1}{4}N'(x)A^2(x-b) + \frac{1}{4}N''(x)A^2(x-b)^2
      +\frac{3}{8}N''(x)N(x)A(x-b)
    + \frac{1}{8}N'(x)^2A(x-b)
    \\
    \nonumber
    &\quad
    +\frac{1}{4}AN''(x)\sigma^2 + \frac{3}{16}N''(x)N'(x)\sigma^2+ \frac{1}{4}N'''(x)A(x-b)\sigma^2 + \frac{3}{16}N'''(x)N(x)\sigma^2
    \\
    &\quad
    + \frac{1}{16}N''''(x)\sigma^4
    \Bigg) + O(h^4).
\end{align}
We now compare this with the conditional expectation under the true distribution. This was calculated up to order $h^2$ in \cite{parameterestimation} and is given by
\begin{align*}
    \E[X_{t_k} \ | \ X_{t_{k-1}}=x] = x + hF(x) + \frac{h^2}{2}\Big(F'(x)F(x) + \frac{1}{2}\sigma^2F''(x)\Big) + O(h^3).
\end{align*}
To expand this up to order $h^3$, we apply the infinitesimal generator $\mathcal{L}$ to the function $g(x)=F'(x)F(x)+\frac{1}{2}\sigma^2F''(x)$, as per Lemma \ref{lemma infinitesimal generator}. This yields
\begin{align*}
    (\mathcal{L}g)(x) &= F(x)g'(x) + \frac{1}{2}\sigma^2g''(x)\\
    &= F''(x)F(x)^2+F'(x)^2F(x) +  \sigma^2F'''(x)F(x)
    + \frac{3}{2}\sigma^2 F''(x)F'(x) + \frac{1}{4}\sigma^4F''''(x),
\end{align*}
so
\begin{align}\label{eq:true moment}
    \nonumber\E & \big[ X_{t_k}  \ | \ X_{t_{k-1}}=x \big]
    =
    x + hF(x) + \frac{h^2}{2}\left(F'(x)F(x)+\frac{1}{2}\sigma^2F''(x)\right)\\
    & + h^3\Bigg( \frac{1}{6}\Big(F''(x)F(x)^2  + F'(x)^2F(x) + \sigma^2F'''(x)F(x)\Big)
    + \frac{1}{4} \sigma^2F''(x)F'(x)
    +\frac{1}{24}\sigma^4F''''(x)
    \Bigg) + O(h^4).
\end{align}
Comparing \eqref{eq: Strang expectation} and \eqref{eq:true moment} and using that $F'''(x)F(x)=N'''(x)(N(x)+A(x-b))$,  $F''(x)F'(x)=N''(X)(A+N'(x))$, $F''(x)F(x)^2 = N''(x)\Big( N(x)^2+A^2(x-b)^2+2N(x)A(x-b)
    \Big)$ and $F'(x)^2F(x) = A^3(x-b) + N'(x)^2A(x-b) + 2N'(x)A^2(x-b)
    +A^2N(x) + N'(x)^2N(x) + 2AN'(x)N(x)$, we obtain 
\begin{align*}
    \E\big [ \Phi_h(X_{t_{k-1}})- &X_{t_k}\,|\, X_{t_{k-1}}=x \big ]   \\
    = 
    &\, h^3 \Bigg(
    \frac{1}{24}AN'(x)N(x) 
    + \frac{1}{12}A^2N(x)
    -\frac{1}{12}N'(x)A^2(x-b)
    +\frac{1}{12}N''(x)A^2(x-b)^2     \\
    &\qquad
    +\frac{1}{24}N''(x)N(x)A(x-b)
    -\frac{1}{24}N'(x)^2A(x-b)
    - \sigma^2\frac{1}{16}N''(x)N'(x)      \\
    &\qquad
    +\sigma^2\frac{1}{12}N'''(x)A(x-b)
    +\sigma^2\frac{1}{48}N'''(x)N(x)
    +\sigma^4\frac{1}{48}N''''(x)
    \Bigg)+O(h^4).
\end{align*}
This concludes the proof of (i).

\ny
(ii) We first derive the true variance. 
Applying Lemma \ref{lemma infinitesimal generator} with $g(x)=x^2$, we get
\begin{align*}
(\mathcal{L}g)(x)= &2xF(x) + \sigma^2, \\
(\mathcal{L}^2g)(x)= &2F(x)^2 + 2xF'(x)F(x)+ 2\sigma^2F'(x)+\sigma^2xF''(x), \\
(\mathcal{L}^3g)(x)= &6F'(x)F(x)^2
       + 2xF''(x)F(x)^2
       + 2xF'(x)^2F(x) +7\sigma^2F''(x)F(x)
       + 2\sigma^2xF'''(x)F(x)\\
       &
       + 4\sigma^2F'(x)^2  + 3\sigma^2xF''(x)F'(x)
       + 2\sigma^4F'''(x)
       + \frac{1}{2}\sigma^4F''''(x).
\end{align*}
This means that the second moment is
\begin{align*}
\E\left[X_{t_k}^2 \mid X_{t_{k-1}} = x\right]
= x^2
   + &h\bigl(2xF(x)+\sigma^2\bigr) \\
\quad + &h^2\biggl(
      F(x)^2+xF'(x)F(x)
      +\sigma^2F'(x)
      +\frac12\sigma^2xF''(x)
   \biggr) \\
\quad + &h^3\biggl(
      F'(x)F(x)^2
      +\frac13xF''(x)F(x)^2
      +\frac13xF'(x)^2F(x) \\
&\quad
      +\frac76\sigma^2F''(x)F(x)
      +\frac13\sigma^2xF'''(x)F(x)
      +\frac23\sigma^2F'(x)^2 \\
&\quad
      +\frac12\sigma^2xF''(x)F'(x)
      +\frac13\sigma^4F'''(x)
      +\frac1{12}\sigma^4xF''''(x)
   \biggr)
   + O(h^4).
\end{align*}
Squaring the true expectation found in \eqref{eq:true moment} yields 
\begin{align*}
\E[X_{t_k}\mid X_{t_{k-1}}=x]^2
= x^2
   + &h\cdot 2xF(x) \\
\quad + &h^2\biggl(
      F(x)^2
      + x F'(x)F(x)
      + \frac{\sigma^2}{2} xF''(x)
   \biggr) \\
\quad + &h^3\biggl(
      F'(x)F(x)^2
      + \frac{\sigma^2}{2}F''(x)F(x)
      + \frac{1}{3}xF''(x)F(x)^2 
      + \frac{1}{3}xF'(x)^2F(x)\\
&\quad
      + \frac{1}{3}\sigma^2xF'''(x)F(x) 
      + \frac{1}{2}\sigma^2xF''(x)F'(x)
      + \frac{1}{12}\sigma^4xF''''(x)
   \biggr)
   +O(h^4)
\end{align*}
and the true variance is then 
\begin{align*}
    \V[X_{t_k}\mid X_{t_{k-1}}=x]&=\E[X_{t_k}^2\mid X_{t_{k-1}}=x]-\E[X_{t_k}\mid X_{t_{k-1}}=x]^2\\
    &=h\sigma^2+h^2\sigma^2F'(x)+h^3\left(\frac{2}{3}\sigma^2F''(x)F(x)+\frac{2}{3}\sigma^2F'(x)^2+\frac{1}{3}\sigma^4F'''(x)\right)
    +O(h^4).
\end{align*}

\ny
We now calculate the variance of the Strang one-step prediction.
To find $\E[ (X^{[S]}_{t_k})^2 \ | \ X^{[S]}_{t_{k-1}}=x]$, we square \eqref{eq: fmuf før expectation} and find
\begin{align*}
    \left(f_{\frac{h}{2}}(x+\widetilde{Q}_h(x) + \xi_h)\right)^2
    &=
     x^2 + \widetilde{Q}_h(x)^2 + \xi_h^2
        +2x\widetilde{Q}_h(x)
        +2x\xi_h
        +2\widetilde{Q}_h(x) \xi_h
        +\frac{h^2}{4}N(x+\widetilde{Q}_h(x)+\xi_h)^2
    \\
    &\quad
        +\frac{h^3}{8}N'(x+\widetilde{Q}_h(x)+\xi_h)
            N(x+\widetilde{Q}_h(x)+\xi_h)^2
    \\&\quad
        +hxN(x+\widetilde{Q}_h(x)+\xi_h)
        +\frac{h^2}{4} xN'(x+\widetilde{Q}_h(x)+\xi_h)
            N(x+\widetilde{Q}_h(x)+\xi_h)
    \\
    &\quad
        +\frac{h^3}{24} x \Bigg( N''(x+\widetilde{Q}_h(x)+\xi_h)
            N(x+\widetilde{Q}_h(x)+\xi_h)^2
    \\
   &\quad
   +N'(x+\widetilde{Q}_h(x)+\xi_h)^2N(x+\widetilde{Q}_h(x)+\xi_h)  \Bigg)
   \\
   &\quad
    +h\widetilde{Q}_h(x) N(x+\widetilde{Q}_h(x)+\xi_h)
    +\frac{h^2}{4}\widetilde{Q}_h(x)N'(x+\widetilde{Q}_h(x)+\xi_h)
    N(x+\widetilde{Q}_h(x)+\xi_h)
    \\&\quad
    +h\xi_hN(x+\widetilde{Q}_h(x)+\xi_h)
        +\frac{h^2}{4}\xi_h 
            N'(x+\widetilde{Q}_h(x)+\xi_h)
            N(x+\widetilde{Q}_h(x)+\xi_h)
    \\
    &\quad
    +\frac{h^3}{24}\xi_h 
        \Bigg(N''(x+\widetilde{Q}_h(x)+\xi_h)
            N(x+\widetilde{Q}_h(x)+\xi_h)^2
    \\
    &\quad
        +N'(x+\widetilde{Q}_h(x)+\xi_h)^2
            N(x+\widetilde{Q}_h(x)+\xi_h)  \Bigg)
    +O(h^4).
\end{align*}
We take the expectation of the above by using the same methods as in the proof of (i), the only difference being that we now need one more term in the expansion of the second moment of $\xi_h$,
$\E[\xi_h^2] = h\sigma^2 + h^2A\sigma^2+h^3\frac{2}{3}A^2\sigma^2+O(h^4).$ This yields
\begin{align*}
    \E \left[ \left(
    f_{\frac{h}{2}}(x+\widetilde{Q}_h(x) + \xi_h)\right)^2   \right]
    &= 
    x^2+h\big( \sigma^2+2xF(x) \big)
    +h^2 \left( F(x)^2+\sigma^2F'(x)+\frac{\sigma^2}{2}F''(x) + xF'(x)F(x) \right)
    \\
    &\quad
    + h^3
    \Bigg(
    N'(x)N(x)^2
    +AF(x)N(x)
    +2 N'(x)N(x)A(x-b)
    +F(x)A^2(x-b)
        \\&\quad\quad\quad
    +\frac{2}{3} \sigma^2 A^2
    + \frac{8}{24} xN''(x)N(x)^2
    + \frac{8}{24} xN'(x)^2N(x)
    +\frac{1}{4} xAN'(x)N(x)
    \\&\quad\quad\quad
    +\frac{1}{2} xA^2N(x)
    +\frac{1}{3} xA^3(x-b)
    +\frac{1}{2} \sigma^2N'(x)^2
    +\frac{5}{4} \sigma^2 N''(x)N(x)
        \\&\quad\quad\quad
    +\frac{1}{2} x F(x)AN'(x)
    +\frac{1}{2} xN''(x)A^2(x-b)^2
    +\frac{3}{4} xN''(x)N(x)A(x-b)
        \\&\quad\quad\quad
    +\frac{1}{2} \sigma^2 x AN''(x)
    +\frac{3}{8} \sigma^2 x N'''(x)N(x)
    +\frac{1}{2} \sigma^2 x N'''(x)A(x-b)
        \\&\quad\quad\quad
    +\frac{3}{8} \sigma^2 x N''(x)N'(x)
    +\frac{1}{8} \sigma^4 x N''''(x)
    +\frac{1}{4} xN'(x)^2 A(x-b)
    \\&\quad\quad\quad
    + N'(x)A^2 (x-b)^2
    +\frac{3}{2} \sigma^2 N''(x)A(x-b)
    +\sigma^2 AN'(x)
    +\frac{1}{2} \sigma^4 N'''(x)
    \Bigg)
    \\&\quad + O(h^4).
\end{align*}
In particular, we see that the terms up to order $h^2$ equal those in the true second moment, so the error of the variance is $O(h^3)$.
We now square \eqref{eq: Strang expectation} and find
\begin{align*}
    \left( \E\left[ X^{[S]}_{t_k} \ | \ X^{[S]}_{t_{k-1}}=x \right] \right)^2
    &=
    x^2 +
    2hxF(x) + h^2F(x)^2 
     + h^2x\left(F'(x)F(x)+\frac{1}{2}\sigma^2F''(x)\right)
     \\
    &\quad
    + h^3F'(x)F(x)^2
    +\frac{h^3}{2}\sigma^2F''(x)F'(x)
    \\
    &\quad
    +h^3 x
    \Bigg(
    \frac{3}{4}AN'(x)N(x) + \frac{1}{2}A^2N(x) + \frac{1}{3}A^3(x-b) + \frac{1}{3}N''(x)N(x)^2
    \\
    \nonumber
    &\quad  
    +\frac{1}{3}N'(x)^2N(x)  +\frac{1}{2}N'(x)A^2(x-b) + \frac{1}{2}N''(x)A^2(x-b)^2
    \\
    \nonumber
    &\quad
    +\frac{3}{4}N''(x)N(x)A(x-b)
    + \frac{1}{4}N'(x)^2A(x-b)
    +\frac{1}{2}AN''(x)\sigma^2 
    \\
    \nonumber
    &\quad
    + \frac{3}{8}N''(x)N'(x)\sigma^2
    + \frac{1}{2}N'''(x)A(x-b)\sigma^2 
    + \frac{3}{8}N'''(x)N(x)\sigma^2
    \\
    &\quad
    + \frac{1}{8}N''''(x)\sigma^4
    \Bigg) + O(h^4).
\end{align*}
The variance of the Strang one-step prediction is thus
\begin{align*}
    \V\Big[ X^{[S]}_{t_k} \ |  \ X^{[S]}_{t_{k-1}}=x \Big]
    =
    \E \Bigg [ \Big( 
    f_{\frac{h}{2}}(x+&\widetilde{Q}_h(x) + \xi_h)\Big)^2   \Bigg] 
    -
    \left( \E\left[ X^{[S]}_{t_k} \ | \ X^{[S]}_{t_{k-1}}=x \right] \right)^2
    \\
    =
    h\sigma^2 + h^2\sigma^2 F'(x)
    +h^3 \Bigg( 
    &\frac{2}{3}\sigma^2A^2
     +\frac{1}{2} \sigma^2N'(x)^2
    +\frac{3}{4} \sigma^2 N''(x)N(x)
      \\ 
    &+ \sigma^2 N''(x)A(x-b)
    +\sigma^2 AN'(x)
    +\frac{1}{2} \sigma^4 N'''(x)
    \Bigg) 
    +O(h^4).
\end{align*}
We conclude that the error of the variance is
\begin{align*}
    \V\Big [ \Phi_h^{[S]}&(X_{t_{k-1}}) \ | \ X^{[S]}_{t_{k-1}}=x\Big]-\V[X_{t_k} \ | \ X_{t_{k-1}}=x]
    =\\
    &h^3
    \Bigg(
    \sigma^2\Big(\frac{1}{12}N''(x)N(x) 
    + \frac{1}{3}N''(x)A(x-b)
    -\frac{1}{6}N'(x)^2
    -\frac{1}{3}AN'(x)\Big)+\frac{\sigma^4}{6}N'''(x) 
    \Bigg)
    + O(h^4).
\end{align*}
\end{proof}

\ny
\begin{proof}[Proof that the FitzHugh-Nagumo and the double-well potential model satisfy the model assumptions] 
    In both models, the drift function is a polynomial so  we only have to show \textbf{A2}. For this it suffices to show that $N$ is one-sided Lipschitz. Let $F$ be the drift of the FitzHugh-Nagumo model with splitting given by \eqref{eq: FHN splitting: linear part A(x-b)}, \eqref{eq: FHN splitting: N(x,y)} for $c=(c_1,c_2)\in\R^2$ and consider similarly the splitting of $\widetilde{F}$, the drift of the double-well potential model, with a linearization around $c_1$ such that $\widetilde{N}(x)=\frac{1}{\e}(-x^3+3c_1^2x-2c_1^3)$.

\ny
Let $x=(x_1,x_2)$ and $y=(y_1,y_2)\in\R^2$. Since $N^{(2)}$  is constantly zero,
\begin{align*}
    (x-y)^{\top} \big(N(x)-N(y)\big)
    &= (x_1-y_1) \big( N^{(1)}(x_1,x_2) - N^{(1)}(y_1,y_2) \big)\\ 
    &= \frac{1}{\e}(x_1-y_1)\big(-x_1^3+y_1^3+3c_1^2(x_1-y_1)\big)\\
    &= \frac{3c_1^2}{\e}(x_1-y_1)^2 - \frac{1}{\e}(x_1-y_1)(x_1^3-y_1^3)\\
    &= \frac{3c_1^2}{\e}(x_1-y_1)^2 - (x_1-y_1)^2(x_1+y_1)^2\\
    &\leq \frac{3c_1^2}{\e} ||x-y||^2,
\end{align*}
so $N$ is one-sided Lipschitz with $C=\frac{3c_1^2}{\e}.$ 

\ny
Since $\widetilde{N}(x)-\widetilde{N}(y)=N^{(1)}(x,x_2)-N^{(1)}(y,y_2)$ for any $x,y\in\R$, it follows immediately from the above calculations that the drift function of the double-well potential model also fulfills the assumptions.
\end{proof}

\subsection*{Funding}
This work was funded by Novo Nordisk Foundation NNF20OC0062958.

\subsection*{Code availability}
Computer code to reproduce all simulations and figures can be found in the following repository: https://github.com/qjb831/Optimal-Strang-splitting-article

\bibliographystyle{plainnat}
\bibliography{kildeliste}

@article{parameterestimation,
  author =       "Predrag Pilipovic and Adeline Samson and Susanne Ditlevsen",
  title =        "Parameter Estimation in Nonlinear Multivariate Stochastic Differential Equations Based on Splitting Schemes",
  journal =      "The Annals of Statistics",
  volume =       "52",
  number =       "2",
  pages =        "842--867",
  year =         "2024"
}

@book{thygesen,
    title={Stochastic Differential Equations for Science and Engineering},
    author={Uffe Høgsbro Thygesen},
    year={2023},
    publisher={CRC Press}
}

@book{kutoyants,
    author = "Yury A. Kutoyants",
    title = "Statistical Inference for Ergodic Diffusion Processes",
    publisher = "Springer",
    year = "2004"
}

@article{buckwar_etal,
    author = "Evelyn Buckwar and Adeline Samson and Massimiliano Tamborrino and Irene Tubikanec",
    title = " {A splitting method for {SDE}s with locally {L}ipschitz drift: 
{I}llustration on the {F}itz{H}ugh-{N}agumo model}",
    journal = "Applied Numerical Mathematics",
    volume = "179",
    pages = "191--220",
    year = "2022"
}

@book{applied,
    author = "Simo Särkkä and Arno Solin",
    title = "Applied Stochastic Differential Equations",
    publisher = "Cambridge University Press",
    year = "2019"
}

@article{FITZHUGH1961445,
title = {{Impulses and Physiological States in Theoretical Models of Nerve Membrane}},
journal = {Biophysical Journal},
volume = {1},
number = {6},
pages = {445-466},
year = {1961},
author = {Richard FitzHugh}
}

@ARTICLE{Nagumo1962,  
author={Nagumo, J. and Arimoto, S. and Yoshizawa, S.},  
journal={Proceedings of the IRE},   
title={{An Active Pulse Transmission Line Simulating Nerve Axon}},   year={1962},  
volume={50},  
number={10},  
pages={2061-2070},  
}

@article{benzi1982stochastic,
  title={Stochastic resonance in climatic change},
  author={Benzi, Roberto and Parisi, Giorgio and Sutera, Alfonso and Vulpiani, Angelo},
  journal={Tellus},
  volume={34},
  number={1},
  pages={10--16},
  year={1982}
}

@article{benzi1983theory,
  title={A theory of stochastic resonance in climatic change},
  author={Benzi, Roberto and Parisi, Giorgio and Sutera, Alfonso and Vulpiani, Angelo},
  journal={SIAM Journal on Applied Mathematics},
  volume={43},
  number={3},
  pages={565--578},
  year={1983},
 publisher = {Society for Industrial and Applied Mathematics}
}

@InBook{EstimatingFunctions,
author = {Sørensen, Michael},
editor = {Kessler, Matthieu and Lindner, Alexander and Sørensen, Michael},
year = {2012},
month = {05},
chapter = {1},
pages = {1--97},
bookTitle = {{Statistical Methods for Stochastic Differential Equations}}, 
title = {{Estimating functions for diffusion-type processes}},
isbn = {978-1-4398-4940-8},
publisher = {Chapman and Hall/CRC},
doi = {10.1201/b12126-2},
}

@article{ditlevsen:2023,
Author = {Ditlevsen, Peter and Ditlevsen, Susanne},
Title = {Warning of a forthcoming collapse of the {Atlantic meridional overturning
   circulation}},
Journal = {Nature Communications},
Year = {2023},
Volume = {14},
Number = {1},
DOI = {10.1038/s41467-023-39810-w},
Article-Number = {4254},
EISSN = {2041-1723},
ResearcherID-Numbers = {Ditlevsen, Susanne/F-6708-2012
   },
ORCID-Numbers = {Ditlevsen, Susanne/0000-0002-1998-2783
   Ditlevsen, Peter/0000-0003-2120-7732},
Unique-ID = {WOS:001037058500010},
}

@article{pilipovic:2025,
title = {Strang splitting for parametric inference in second-order stochastic differential equations},
journal = {Stochastic Processes and their Applications},
volume = {187},
pages = {104650},
year = {2025},
issn = {0304-4149},
doi = {https://doi.org/10.1016/j.spa.2025.104650},
url = {https://www.sciencedirect.com/science/article/pii/S0304414925000912},
author = {Predrag Pilipovic and Adeline Samson and Susanne Ditlevsen},
}

@techreport{ditlevsen:2026,
    author = {Ditlevsen, Susanne and Rahbek, Anders and Damsholt, Gabriel and Ditlevsen, Peter},
    title = {Estimating the critical level of {CO2} causing a collapse of the {Atlantic Meridional Overturning Circulation}},
    institution = {PREPRINT, Research Square},
    year = {2026},
    doi = {https://doi.org/10.21203/rs.3.rs-9180655/v1}
}

@article{MCMC,
    author = {Anders C. Jensen and Susanne Ditlevsen and Mathieu Kessler and Omiros Papaspiliopoulos},
    title = {{Markov chain Monte Carlo approach to parameter estimation in the FitzHugh-Nagumo model}},
    journal = "Physical Review E",
    year = "2012",
    volume = "86",
    note = "Article 041114"
}

@book{dynamics,
    author = "Steven H. Strogatz",
    title = "Nonlinear Dynamics and Chaos: With Applications to Physics, Chemistry, and Engineering",
    publisher = "CRC Press",
    year = "2018",
    edition = "2."
}

@book{berglund,
    author = "Nils Berglund and Barbara Gentz",
    title = "Noise-Induced Phenomena in Slow-Fast Dynamical Systems",
    publisher = "Springer",
    year = "2006"
}


\end{document}